\documentclass[11pt]{article}
\usepackage{fullpage}
\usepackage{amsthm,amsfonts}
\usepackage{amsmath}
\allowdisplaybreaks[4]
\usepackage{amsthm}
\usepackage{listings} % 引入代码列表包
\usepackage{xcolor}   % 引入颜色包
\usepackage{caption}
\usepackage{tikz}
\usetikzlibrary{arrows.meta}

\usepackage{algorithm}
\usepackage{algpseudocode}
\usepackage{amssymb}
\usepackage{color}
\usepackage{mathrsfs}
\usepackage{enumitem}
\usepackage{bm}
\usepackage{multirow}
\usepackage{booktabs}
\usepackage{makecell}
\usepackage{graphicx}
\usepackage{comment}
\usepackage{cases}
\usepackage{appendix}
\usepackage{tikz}
\usetikzlibrary{arrows,shapes}
\usepackage{marginnote}
\usepackage{tcolorbox}
\usepackage[qm]{qcircuit}
\usepackage{braket}

\colorlet{color1}{blue}
\colorlet{color2}{red!50!yellow}

\usepackage{hyperref}[6.83]
\hypersetup{
  colorlinks=true,
  frenchlinks=false,
  pdfborder={0 0 0},
  naturalnames=false,
  hypertexnames=false,
  breaklinks,
  allcolors = color1,
  urlcolor = color2,
}

\usepackage[capitalise]{cleveref}

\numberwithin{equation}{section}

\usepackage{graphicx}
\usepackage{anyfontsize}
\usepackage{amsmath}
\usepackage{xcolor}
\usepackage{subcaption}
\usepackage{caption}
\usepackage{graphicx}
\usepackage{booktabs}
\usepackage{pifont}
\usepackage{amssymb}
\usepackage{algorithm}
\usepackage{algpseudocode}
\usepackage{amsfonts}
\usepackage{mathrsfs}

\newtheorem{theorem}{Theorem}
\newtheorem{lemma}[theorem]{Lemma}
\newtheorem{proposition}[theorem]{Proposition}

\newtheorem{definition}[theorem]{Definition}
\newtheorem{result}[theorem]{Result}

\allowdisplaybreaks
\title{Contour-integral and Fourier transform based multivariable quantum eigenvalue transformation for commuting matrices}
\author{Shan Jiang$^{1}$, \quad Dong An$^{2}$ \\ 
\footnotesize $^{1}$ School of Mathematical Sciences, Peking University, Beijing, China\\
\footnotesize $^{2}$ Beijing International Center for Mathematical Research, Peking University, Beijing, China\\
}

\date{ }

\begin{document}

\maketitle

\begin{abstract}
We study the problem of implementing multivariable matrix-valued functions on a quantum computer and propose two quantum algorithms for multivariable matrix eigenvalue transformations acting on tuples of pairwise commuting matrices. 
The first algorithm is based on multivariable contour integrals and applies to arbitrary holomorphic functions on $\mathbb{C}^n$. 
The second algorithm is based on high-dimensional Fourier transforms and designed for smooth functions of commuting Hermitian matrices. 
For both algorithms, we discuss and analyze the complexity of their quantum implementation leveraging quantum singular value transformation, compression gadgets, and linear combination of unitaries. 
We additionally study their hybrid quantum--classical variants that reduce the number of required ancilla qubits to logarithmic dependence on the number of variables, at the cost of increased query complexity. As an application, we show that multivariable matrix polynomials can be implemented with no additional explicit degree dependence, depending only on global properties of the polynomials.
\end{abstract}

\tableofcontents
\section{Introduction}
Computing matrix-valued functions is a crucial problem in scientific computing. Many practical problems can be directly represented as matrix-valued functions, such as matrix inversion, Hamiltonian simulation and the simulation of dynamical systems through matrix exponentials. 
Furthermore, the matrix-valued function problem can also act as a subproblem to solve more complex problems. 
For example, matrix $\varphi$-function is a fundamental subroutine for nonlinear time integration, 
two-matrix geometric means arise naturally in Wasserstein barycenter computation, and matrix square roots and inverse square roots are widely used in covariance estimation and statistical computations. 
Despite its wide applications, implementing matrix-valued functions, especially for high-dimensional matrices, generally remains computationally expensive on classical computers. 

In recent years, quantum computers have demonstrated the potential to accelerate matrix-valued function problems through quantum superposition and entanglement.
For this task, the quantum singular value transformation (QSVT)~\cite{GilyenSuLowEtAl2018} is a remarkable framework. 
Based on quantum signal processing (QSP)~\cite{LowChuang2017}, QSVT provides a unified method for problems that can be represented as \emph{singular value transformation} of matrices both time- and space-efficiently, and has been demonstrated to unify major quantum algorithms including search, phase estimation and Hamiltonian simulation~\cite{MartynRossiTanEtAl2021}. 
Nevertheless, although certain scientific computing problems such as matrix inversion can be represented as singular value transformation, many other problems, including matrix polynomials and solving differential equations, are naturally \emph{eigenvalue transformation} instead of singular value transformation. 
This inspires recent rapid progresses on quantum algorithms for matrix eigenvalue transformation, such as quantum eigenvalue processing~\cite{LowSu2024}, Laplace transform based quantum eigenvalue transformation~\cite{AnChildsLinYing2024}, contour integral approach~\cite{TakahiraOhashiSogabeEtAl2020,TakahiraOhashiSogabeEtAl2021,JiangAn2026}, and Poisson summation approach~\cite{WangZhuangDouEtAl2026}. 
For either singular value transformation or eigenvalue transformation, these existing algorithms mainly focus on computing single-variable matrix functions, i.e., functions of a single matrix. 

In matrix-valued function problem, the computation of multivariable matrix functions, i.e., functions of multiple matrices is another important task with broad applications. 
For example, digital simulation of time-dependent evolutions can be reduced to multivariable matrix-valued function problems, such as the time-dependent Hamiltonian simulation or ordinary differential equation evolution.
In multigraph filtering, different edge types are represented by multiple graph shift matrices, and the resulting filtering operation is described by a multivariable matrix polynomial. 
Despite its wide applications, a general quantum algorithm framework for multivariable matrix functions is still lacking, and it can be challenging to generalize algorithms for single-variable matrix functions to the multivariable case. 
In QSP/QSVT framework, several recent works~\cite{RossiChuang2022,NemethKoverEtAl2023, MoriMizutaEtAl2024, LaneveWolf2024, Laneve2025} have attempted to extend QSP to multivariable case, whose success may lead to a multivariable matrix singular value transformation algorithm. 
However, compared to the completeness of univariate QSP, the conditions required for multivariable approaches remain more stringent, making it difficult to handle general multivariable QSP and QSVT. 
Additionally, for general eigenvalue transformations, quantum algorithms based on
polynomial approximation, univariate QSP, and LCU have been
developed towards commuting Hermitian matrices~\cite{bornsweil2023mqet}, but in the absence of the Hermitian assumption, efficient quantum algorithms for the multivariable case are still missing. 

In this work, we make a step towards the realization of general multivariable matrix-valued functions by designing two efficient quantum multivariable eigenvalue transformation algorithms for \emph{mutually commuting matrices}. 
The first approach is based on contour integral, which applies any holomorphic function on $\mathbb{C}^n$ to a set of mutually commuting matrices, and the second is based on Fourier transform, which applies any function $f$ satisfying $f(\boldsymbol{x})\in H^s(\mathbb{R}^n)\cap L^1(\mathbb{R}^n)$, $s>\frac{n}{2}$ and $x_jf(\boldsymbol{x})\in L_1(\mathbb{R}^n)$ to a set of mutually commuting Hermitian matrices. 
Our quantum algorithms are realized from these two representations by the linear combination of unitaries (LCU) subroutine~\cite{ChildsWiebe2012}. 
Additionally, we combine these approaches and the recently proposed sampling-based LCU~\cite{Chakraborty2024,WangMcArdleBerta2024} to propose qubit-efficient multivariable quantum eigenvalue transformation algorithms. 

We carefully analyze the general complexities of our algorithms, and study their efficiency for computing matrix polynomials as a concrete application. 
Specifically, the complexities of our algorithms for polynomials depend on analytical properties such as the polynomial's maximum value and Lipschitz continuity, rather than the polynomial's degree or the coefficients of each term in the polynomial. 
This is highly advantageous compared to a straightforward LCU when implementing certain polynomials such as Chebyshev polynomials, which are well-behaved but with exponentially large $1$-norm of the coefficients.

\subsection{Problem statement}
Let $\mathcal{A} = \mathbb{C}^{m \times m}$ denote the space of complex $m \times m$ matrices. Let $|\psi\rangle$ be an $m$-dimensional quantum state, and let
\[
\boldsymbol{A} = (A_1, A_2, \ldots, A_n) \in \mathcal{A}^n
\]
be an $n$-tuple of pairwise commuting matrices. Our goal is to design quantum algorithms that approximate the action of a multivariable matrix eigenvalue transformation
\[
f(\boldsymbol{A}) = f(A_1, A_2, \ldots, A_n)
\]
on the input state $|\psi\rangle$, where $f : \mathbb{C}^n \to \mathbb{C}$ is a complex-valued function. Here we define this problem as below \cite{HartmannLesch2024}:

\begin{definition}[Multivariable matrix eigenvalue transformation]
\label{definition:Multi-matrix_function}
Let $\mathbb{F}\in\{\mathbb{C},\mathbb{R}\}$ be a scalar field. For an open set $\Omega\in\mathbb{F}^n$, let $f:\mathbb{F}^n\rightarrow\mathbb{F}$ be a multivariable function on $\Omega$. Let $\mathcal{A}=\mathbb{C}^{m\times m}$ be the matrix space, and $\boldsymbol{A}=\left(A_{1}, \cdots, A_{\mathrm{n}}\right) \in \mathcal{A}^{\mathrm{n}}$ be an $n$-tuple of commuting matrices with its joint spectrum $\operatorname{spec}_{\mathcal{A}}\left(\boldsymbol{A}\right)\subset\Omega\subset \mathbb{F}^n$. Then $f(\boldsymbol{A})$ is a multi-matrix eigenvalue transformation (or simply called multi-matrix function in this work) if the following proposition holds.
\begin{enumerate}
\item For a polynomial $f \in \mathbb{F}\left[\mathrm{z}_{1}, \cdots, \mathrm{z}_{\mathrm{n}}\right]$, of the form
$$
f(\boldsymbol{z})=f\left(z_{1}, \cdots, z_n\right)=\sum_{\alpha \in \mathbb{N}^n} \lambda_{\alpha} \cdot z_{1}^{\alpha_{1}} \cdots z_n^{\alpha_n}=: \sum_{\alpha \in \mathbb{N}^n} \lambda_{\alpha} \cdot z^{\alpha}, 
$$
we have
$$
f(\boldsymbol{A})=f\left(A_{1}, \cdots, A_{n}\right)=\sum_{\alpha \in \mathbb{N}^{n}} \lambda_{\alpha} \cdot A_{1}^{\alpha_{1}} \cdots A_{n}^{\alpha_{n}}=: \sum_{\alpha \in \mathbb{N}^{n}} \lambda_{\alpha} \cdot A^{\alpha}.
$$
\item For $f:\mathbb{F}^n\rightarrow\mathbb{F}$, if $r>n$, $A_{\mathrm{n}+1}, \cdots, A_{\mathrm{r}} \in \mathcal{A}$ and $\widetilde{f}:\mathbb{F}^r\rightarrow\mathbb{F}$ is the trivial extension of $f$ defined by $\widetilde{f}\left(z_{1}, \cdots, z_{\mathrm{r}}\right)=f\left(z_{1}, \cdots, z_{\mathrm{n}}\right)$. Then denoting $\widetilde{\boldsymbol{A}}:=\left(A_{1}, \cdots, A_{\mathrm{r}}\right)$, we have
$$
\widetilde{f}(\widetilde{\boldsymbol{A}})=f(\boldsymbol{A}).
$$
\item If $\left\{f_k\right\}_{k \in \mathbb{N}}$ is a sequence of functions which converges locally uniformly to $f \in \mathbb{F}^n$, then $f_{k}(\boldsymbol{A})\rightarrow f(\boldsymbol{A})$.
\end{enumerate}

\end{definition}

When these commuting matrices can be diagonalized by the same set of invertible matrices, let $A_j=QD_jQ^{-1}$ for $j=1,\cdots,n$, where $Q$ is an invertible matrix and $\boldsymbol{D}=\left(D_{1}, \cdots, D_{\mathrm{n}}\right)$ is an $n$-tuple of diagonal matrices. Suppose $D_j=\operatorname{diag}(d_{j1},\cdots,d_{jm})$, then it is straightforward to show that $f(\boldsymbol{A})=Qf(\boldsymbol{D})Q^{-1}$ is a multi-matrix function that satisfies \cref{definition:Multi-matrix_function}, where $f(\boldsymbol{D})=\operatorname{diag}(f(\boldsymbol d_1),\cdots,f(\boldsymbol d_m))$ and $\boldsymbol{d}_k=(d_{1k},\cdots,d_{nk})$. 
Therefore, \cref{definition:Multi-matrix_function} is a self-consistent way of defining eigenvalue transformation for mutually commuting matrices. 

For each $j = 1, \ldots, n$, assume that $\alpha_j \ge \|A_j\|$ and that there exists an $(\alpha_j, a, 0)$ block-encoding of $A_j$, denoted by $U_{A_j}$, acting on $a$ ancilla qubits, such that
\[
U_{A_j} =
\begin{pmatrix}
A_j / \alpha_j & \cdot \\
\cdot & \cdot
\end{pmatrix}.
\]
Furthermore, assume access to a state-preparation oracle $U_\psi$ satisfying
\[
U_\psi |0\rangle = |\psi\rangle .
\]
Using the oracles $\{U_{A_j}\}_{j=1}^n$ and $U_\psi$, we formalize the following computational tasks.

\begin{definition}[Quantum Matrix Function Problem]
\label{definition:main_problem}
Let $A_1, \ldots, A_n \in \mathbb{C}^{m \times m}$ be mutually commuting matrices with block-encodings $\{U_{A_j}\}_{j=1}^n$, and let $|\psi\rangle$ be a quantum state representing an $m$-dimensional complex vector. For a matrix function $f(\boldsymbol{A})$, define the normalized target state
\[
|f\rangle := \frac{f(\boldsymbol{A}) |\psi\rangle}{\| f(\boldsymbol{A}) |\psi\rangle \|}.
\]
Given an error parameter $0 \le \epsilon \le \tfrac{1}{2}$, the goal is to output a quantum state $|\widetilde{f}\rangle$ such that
\[
\left\|\, |f\rangle - |\widetilde{f}\rangle \,\right\| \le \epsilon
\]
with probability at least $1/2$.
\end{definition}

In certain applications, it is unnecessary to explicitly prepare the quantum state $|f\rangle$. Instead, one is interested only in estimating its expectation value with respect to a given observable $O$. This leads to the following estimation problem.

\begin{definition}[Quantum Matrix Function Estimation Problem]
\label{definition:estimation_problem}
Let $A_1, \ldots, A_n \in \mathbb{C}^{m \times m}$ be matrices with block-encodings $\{U_{A_j}\}_{j=1}^n$, and let $|\psi\rangle$ be a quantum state. Given a matrix function $f(\boldsymbol{A})$, an observable $O$, and accuracy parameters $\epsilon > 0$ and $\delta \in (0,1)$, the task is to output a real number $\mu$ such that
\[
\left| \mu - \langle \psi | f^\dagger(\boldsymbol{A})\, O\, f(\boldsymbol{A}) | \psi \rangle \right| \le \epsilon
\]
with probability at least $1 - \delta$.
\end{definition}

\subsection{Main results}

\subsubsection{Quantum multivariable contour integral algorithm}
In this method we consider $\mathbb{F}=\mathbb{C}$ and $f:\mathbb{C}^n\rightarrow\mathbb{C}$ is a holomorphic function on $\Omega$.
The quantum multivariable contour integral algorithm is based on the integral formula:
\begin{equation}
f(\boldsymbol{A})=f\left(A_{1}, \cdots, A_{n}\right):=\frac{1}{(2 \pi \mathrm{i})^n}\int_{\Gamma_{1}} \cdots \int_{\Gamma_{n}} f(\boldsymbol{z}) \cdot \prod_{j=1}^{n}\left(z_{j}I_m-A_{j}\right)^{-1} \mathrm{~d} \boldsymbol{z},
\end{equation}
here $\boldsymbol{z}=(z_{1}, \cdots, z_{n})$ and $\mathrm{d} \boldsymbol{z}=\mathrm{d} z_{1} \cdots \mathrm{~d} z_{n}$, and $\Gamma_j, 1 \leq \mathrm{j} \leq \mathrm{n}$, are integration contours that encircle $\operatorname{spec}_{\mathcal{A}}\left(A_j\right)$. This formula has been proved in \cite[Theorem 2.4]{HartmannLesch2024} to satisfy \cref{definition:Multi-matrix_function}. Similar to \cite{TakahiraOhashiSogabeEtAl2021,JiangAn2026}, we implement the matrix inverse $(z_jI_m-A_j)^{-1}$ by QSVT, use a compression gadget to multiply them, and  use quantum LCU~\cite{ChildsWiebe2012} to compute the finite sum from the integral discretization to approximate the matrix function in our target.

Our first main result is the complexity analysis of this algorithm.

\begin{result}[Informal version of \cref{theorem:contour_integral_complexity}]
Ignoring the $n$-th power of some constant, the quantum multivariable contour integral algorithm can solve the problem defined by \cref{definition:main_problem} using 
$\widetilde{\mathcal{O}}\left(\frac{B\prod_{j=1}^n\gamma_j l_j}{\left\|f(\boldsymbol A)\ket{\psi}\right\|}\cdot\gamma_j\alpha_j\mathrm{log}\frac{1}{\epsilon}\right)$ queries to the input model of the matrix $A_j$, $\mathcal{O}\left(\frac{B\prod_{j=1}^n\gamma_j l_j}{\left\|f(\boldsymbol A)\ket{\psi}\right\|}\right)$ queries to the input model of $\ket{\psi}$ and $\mathcal{O}\left(\sum_{j=1}^n\log\frac{n(B\gamma_j + BK_{\Gamma_j}+L)\prod_{j=1}^n l_j\gamma_j}{\|f(\boldsymbol A)\ket{\psi}\|\epsilon}\right)$ additional ancilla qubits. Here $\alpha_j \geq \|A_j\|$, $B=\max_{\boldsymbol z\in\prod_{j=1}^n\Gamma_j}|f(\boldsymbol z)|$, $L$ is the Lipschitz constant of $f$ on the holomorphic area of $f$, $l_j$ is the length of $\Gamma_j$, $K_{\Gamma_j}$ is an upper bound on the tangent variation rate of $\Gamma_j$, and $\gamma_j \geq \|(z_jI_m-A_j)^{-1}\|$ for all $z \in \Gamma_j$.
\label{res:contour}
\end{result}

This result shows that the number of ancilla used in quantum multivariable contour integral algorithm depends logarithmically on many technical parameters. Since these parameters involve $n$-term product, they ultimately exhibit an approximately quadratic dependence on $n$. For the observable estimation problem defined in Definition~\ref{definition:estimation_problem}, we develop a hybrid quantum-classical algorithm whose ancilla-qubit requirement scales only logarithmically with $n$ and is independent of all other technical parameters. The algorithm combines contour integral representation with the sampling-based LCU technique introduced in~\cite{Chakraborty2024}. Related ideas have previously been employed to design qubit-efficient hybrid quantum–classical algorithms for a variety of tasks~\cite{WanBertaCampbell2022,DongLinTong2022,AnLiuLin2023,WangMcArdleBerta2024,WatsonWatkins2024}.
Specifically, as the matrix function can be approximated by the linear combination of the products of matrix inverses, the observable $\bra{\psi}f^\dagger(\boldsymbol A) Of(\boldsymbol A)\ket{\psi}$ can likewise be approximated by a linear combination of observables involving only $O$ and two products of matrix inverses with form $ \prod_{j=1}^{n}\left(z_{j}I_m-A_{j}\right)^{-1} $ with different choices of the parameter vector $\boldsymbol z$. 
We then employ classical sampling to draw values of $\boldsymbol z$ according to a probability distribution determined by the coefficients in the linear combination. For each sampled term, the corresponding observable is independently estimated on a quantum computer using QSVT and compression gadget to implement matrix inverses and the Hadamard test to evaluate the expectation value of $O$. The final estimate is obtained by classically aggregating the sampled outcomes.

The complexity of our algorithm is given in the following result, which shows that in our algorithm the amount of ancilla qubits has logarithmic dependence on $n$, while the overall query complexity has been worse compared to the purely quantum contour integral algorithm. 

\begin{result}[Informal version of \cref{theorem:Ancilla-efficient_complexity_contour}]
Ignoring the $n$-th power of some constant, the quantum multivariable contour integral algorithm can solve the problem defined by \cref{definition:estimation_problem} using
$\widetilde{\mathcal{O}} \left(\frac{\|O\|^2 \gamma_j\alpha_j(B\prod_{j=1}^n\gamma_j l_j)^4}{\epsilon^2 }\right)
$ queries to the input model of the matrix $A_j$, $\mathcal{O} \left(\frac{\|O\|^2 (B\prod_{j=1}^n\gamma_j l_j)^4}{\epsilon^2 }\right)
$queries to the input model of $\ket{\psi}$ and $\mathcal{O}\left(\log n\right)$ additional ancilla qubits. Here
the parameters are defined the same as in \cref{res:contour}.
\end{result}

\subsubsection{Quantum multivariable Fourier transform algorithm}
In this method we consider $\mathbb{F}=\mathbb{R}$ and $f:\mathbb{R}^n\rightarrow\mathbb{R}$ with some smoothness requirements.
When the pairwise commuting matrices are all Hermitian, write $\boldsymbol{H}=(H_1,\cdots,H_n)$, we prove that we can design the quantum multivariable Fourier transform algorithm based on the high-dimensional Fourier transform formula:

\begin{equation}
    f(\boldsymbol{H})=f(H_1,\cdots,H_n)=\lim _{M \rightarrow+\infty} \frac{1}{(2 \pi)^{\frac{n}{2}}} \int_{|\boldsymbol{\lambda}|_\infty \leq M} \widehat{f}(\boldsymbol{\lambda}) \mathrm{e}^{\mathrm{i} \boldsymbol{\lambda} \cdot \boldsymbol{H}} \mathrm{~d} \boldsymbol{\lambda}
\end{equation}

Here $\boldsymbol{\lambda}=(\lambda_1,\lambda_2,\cdots,\lambda_n)\in\mathbf{R}^n$ and $\widehat{f}(\boldsymbol{\lambda})$ is the Fourier transformation of $f(\boldsymbol{x})$. Following the previous idea, we implements the Hamiltonian simulation $\mathrm{e}^{\mathrm{i}\lambda_j H_j}$ by QSVT, use a compression gadget to multiply them, and  use quantum LCU~\cite{ChildsWiebe2012} to compute the finite sum from the integral discretization to approximate the matrix function in our target.

Similarly, we give the complexity analysis of this algorithm as a main result.

\begin{theorem}[Informal version of \cref{theorem:fourier_transform_complexity}]
Let $U_{H_j}$ be an $(\alpha_j,a,0)$-block-encoding of $H_j$ for $j=1,\cdots,n$. Under the assumption that $f(\boldsymbol{x})\in H^s(\mathbb{R}^n)\cap L^1(\mathbb{R}^n)$, $s>k+\frac{n}{2}$ and $x_jf(\boldsymbol{x})\in L_1(\mathbb{R}^n)$ for $j=1,\cdots,n$, let $|f(\boldsymbol{x})|\leq B$, we can solve the problem defined by Definition \ref{definition:main_problem} using a quantum circuit with
$$\widetilde{\mathcal{O}}\left(\left(\alpha_jM+\log (1 / \epsilon^\prime)\right)\cdot \frac{\left(
\pi^{n/2}
\frac{\Gamma\!\left(s-\frac{n}{2}\right)}
{\Gamma(s)}
\right)^{1/2}
\|f\|_{H^s(\mathbb{R}^n)}}{(2 \pi)^{\frac{n}{2}}\|f(\boldsymbol H)\ket{\psi}\|}\right)$$ queries to the input model of the matrix $H_j$, 
$$\mathcal{O}\left(\frac{\left(
\pi^{n/2}
\frac{\Gamma\!\left(s-\frac{n}{2}\right)}
{\Gamma(s)}
\right)^{1/2}
\|f\|_{H^s(\mathbb{R}^n)}}{(2 \pi)^{\frac{n}{2}}\|f(\boldsymbol H)\ket{\psi}\|}\right)$$ queries to the input model of $\ket{\psi}$ and $n\log N+\log n+a+4$ ancilla qubits, where $\epsilon^\prime,M,N$ is given by: 
\begin{equation}
    M\geq\left(\frac{1}{(2 \pi)^{\frac{n}{2}}}\cdot\frac{8}{\|f(\boldsymbol H)\ket{\psi}\|\epsilon}\cdot\sum_{j=1}^n\left\|\left(\frac{\partial^kf}{\partial x_j^k}\right)^{\wedge}\right\|_{L^1(\mathbb{R}^n)}\right)^\frac{1}{k},
\end{equation}
\begin{equation}
    N\geq\frac{(2M)^{n+1}}{(2 \pi)^{n}}\cdot\left(\left\|f\right\|_{L_1} \sum_{j=1}^n\|H_j\|+ \sum_{j=1}^n\left\|x_jf\right\|_{L_1}\right)\cdot\frac{8}{\|f(\boldsymbol H)\ket{\psi}\|\epsilon}.
\end{equation}
\begin{equation}
    \epsilon^\prime\leq\frac{ (2 \pi)^{\frac{n}{2}}\|f(\boldsymbol H)\ket{\psi}\|\epsilon}{2n\left(
\pi^{n/2}
\frac{\Gamma\!\left(s-\frac{n}{2}\right)}
{\Gamma(s)}
\right)^{1/2}
\|f\|_{H^s(\mathbb{R}^n)}}
\end{equation}
\label{res:Fourier}
\end{theorem}

Although the query complexity contains terms exponentially depend of $n$, we show in  \cref{Appendix:Non-exponential dependence of Fourier transform} that when $f$ behaves well they cancel each other out, so the final complexity does not depend exponentially on $n$. Meanwhile, the number of ancilla used ultimately exhibit an approximately quadratic dependence on $n$. Similar to the previous analysis, we also develop a hybrid quantum-classical algorithm whose ancilla-qubit requirement scales only logarithmically with $n$ and is independent of all other technical parameters for the observable estimation problem defined in Definition~\ref{definition:estimation_problem}.
The algorithm is almost identical to the one in the previous part, but products of matrix inverses with form $ \prod_{j=1}^{n}\left(z_{j}I_m-A_{j}\right)^{-1}$ is replaced with products of Hamiltonian simulation with form $\prod_{j=1}^{n} \mathrm{e}^{\mathrm{i}\lambda_jH_j}$.

The complexity of our algorithm is summarized in the following result. It also shows that, the number of required ancilla qubits scales logarithmically with $n$, whereas the overall query complexity is higher than that of the purely quantum contour-integral algorithm.

\begin{theorem}[Informal version of \cref{theorem:Ancilla-efficient_complexity_fourier}]
The quantum multivariable Fourier transform algorithm can solve the problem defined by \cref{definition:estimation_problem} using 
$$
\widetilde{\mathcal{O}} \left(\left(\alpha_jM+\log (1 / \epsilon^\prime)\right)\frac{\|O\|^2 \ln (1 / \xi_2)\left(\frac{1}{(2 \pi)^{\frac{n}{2}}} \left(
\pi^{n/2}
\frac{\Gamma\!\left(s-\frac{n}{2}\right)}
{\Gamma(s)}
\right)^{1/2}
\|f\|_{H^s(\mathbb{R}^n)}\right)^4}{\epsilon^2 }\right)
$$ queries to the input model of matrix $H_j$,
$$
\mathcal{O} \left(\frac{\|O\|^2 \ln (1 / \xi_2)\left(\frac{1}{(2 \pi)^{\frac{n}{2}}} \left(
\pi^{n/2}
\frac{\Gamma\!\left(s-\frac{n}{2}\right)}
{\Gamma(s)}
\right)^{1/2}
\|f\|_{H^s(\mathbb{R}^n)}\right)^4}{\epsilon^2 }\right)
$$ queries to the input model of $\ket{\psi}$ and $\log n+4$ additional ancilla qubits
Here the parameters are defined the same as in \cref{res:Fourier}.
\end{theorem}
\subsubsection{Application in matrix polynomial}
\paragraph{Contour integral}
Let $f(\boldsymbol{z})$ be a multivariable complex polynomial of $\boldsymbol{z}$, and we would like to implement the matrix polynomial $f(\boldsymbol{A})$, which is naturally an eigenvalue transformation. 
For simplicity, we assume $A_j = S D_j S^{-1}$ is diagonalizable and let $\kappa_S = \|S\|\|S^{-1}\|$ denote the condition number of $S$.  
We can choose the contour to be a circle centered at the origin with radius larger than the spectral radius of $A_j$. 
Then, ignoring the $n$-th power of some constant,  \cref{res:contour} shows that the complexity of the contour integral approach is 
\begin{equation}
    \widetilde{\mathcal{O}}\left(\frac{B\kappa_{S}^{n+1}}{\left\|f(\boldsymbol A)\ket{\psi}\right\|}\mathrm{log}\frac{1}{\epsilon}\right). 
\end{equation}
Here $B=\max_{\boldsymbol z\in\prod_{j=1}^n\Gamma_j}|f(\boldsymbol z)|$. 

\paragraph{Fourier transform}
Let $f(\boldsymbol{x})$ be a multivariable real polynomial of $\boldsymbol{x}$. Since $f(\boldsymbol{x})$ doesn't support on bounded regions, we can define a $\mathbb{C}^\infty$ kernel function $k(\boldsymbol{x})$ that takes the value $1$ on the joint spectrum of $f(\boldsymbol x)$, and rapidly decreases to $0$ outside the spectrum of $f(\boldsymbol x)$. Let $g(\boldsymbol{x})=f(\boldsymbol{x})k(\boldsymbol{x})$, \cref{res:Fourier} shows that the query complexity of $H_j$ in the Fourier transform approach is
$$\widetilde{\mathcal{O}}\left(\left(\alpha_jM+\log (1 / \epsilon)\right)\cdot \frac{\left(
\pi^{n/2}
\frac{\Gamma\!\left(s-\frac{n}{2}\right)}
{\Gamma(s)}
\right)^{1/2}
\|g\|_{H^s(\mathbb{R}^n)}}{(2 \pi)^{\frac{n}{2}}\|f(\boldsymbol H)\ket{\psi}\|}\right)$$ with 
$$
M\geq\left(\frac{1}{(2 \pi)^{\frac{n}{2}}}\cdot\frac{8}{\|f(\boldsymbol H)\ket{\psi}\|\epsilon}\cdot\sum_{j=1}^n\left\|\left(\frac{\partial^kg}{\partial x_j^k}\right)^{\wedge}\right\|_{L^1(\mathbb{R}^n)}\right)^\frac{1}{k}
$$
for any positive integer $k$.
\subsection{Related works}

This work is an extension of our previous work \cite{JiangAn2026}. In \cite{JiangAn2026}, we analyzed the complexity of the contour integral method in the one-dimensional case, then we can ask whether this method can be generalized to higher dimensions. Luckily, \cite{HartmannLesch2024} provides us with the necessary integration formula, which allows our work to be extended to higher-dimensional commuting cases with almost no resistance.

In addition to contour integrals, another question is whether there are other similar integral schemes that can transform general high-dimensional matrix-valued functions into integrals of some easily implemented matrix-valued functions. \cite{RosenkranzEtAl2025} uses Fourier transform and Chebyshev transform for quantum state preparation for multivariate functions, and we applied Fourier transform to the problem of multivariable matrix-valued functions. Based on a simple analysis, we think that the Chebyshev transform method may not offer an asymptotic improvement in complexity compared to the Fourier transform method. Therefore, we only provide a detailed analysis of the Fourier transform method in this work. Based on our knowledge, contour integral method is the first to implement multivariable commuting matrix functions that is not limited to a specific function form, but only the holomorphic requirement. For commuting Hermitian matrices construction methods utilizing polynomial approximation and QSP have been developed~\cite{bornsweil2023mqet}.
Here we develop a complementary Fourier-based algorithms that accommodate Hermitian commuting matrices and provide complexity bounds in terms of analytic properties of the target function. It characterizes the cost through Fourier norms and frequency cutoffs. It provide an alternative to polynomial-degree-based estimates; their relative advantages depend on the function and the relevant analytic norms.

Besides our work, a series of works based on QSP and QSVT have the potential to realize multivariable matrix value functions. This path may lead to singular value transformations of multivariable matrices, which are slightly different from our eigenvalue transformations of multivariable matrices.
\cite{RossiChuang2022} provided a multivariable quantum signal processing method (MQSP). However, this method has a limitation on the multivariable polynomials it can represent, making it difficult to provide quantum algorithms for general polynomials like QSP. Subsequent work \cite{NemethKoverEtAl2023, LaneveWolf2024} further explored the polynomials that MQSP could express, but still did not overcome this problem.

\subsection{Discussions}

One limitation of our work is that it can only implement commutative multivariable matrix-valued functions. For non-commutative matrices, rewriting matrix-valued functions as scalar-valued functions results in the loss of information about the order of variables. Therefore, a method similar to the integral transform-based approach used in this paper for non-commutative matrices may not exist. We might be able to perform a Taylor expansion on this function and then directly sum all the expanded terms using LCU, but the advantages of this approach are still worth exploring.

Furthermore, we must acknowledge that most of the simple multivariable matrix-valued functions we encounter can be implemented by separately implementing the univariate matrix-valued functions and then combining them. For complex multivariate matrix-valued functions, our algorithm struggles to achieve good results due to the potential exponential dependence of the dimension $n$. This also limits the application scenarios of our method somewhat.
\subsection{Organization}

The rest of this work is organized as follows. 
In Section \ref{sec:prelim}, we briefly introduce some preliminary results needed in our analysis. 
Then we discuss contour integral algorithm in Section \ref{sec:contour}, and Fourier transform algorithm in Section \ref{sec:fourier}. 
We show the applications in Section \ref{sec:application}. 

\section{Preliminaries}\label{sec:prelim}

\subsection{LCU lemma}
The linear combination of unitaries (LCU) lemma provides a standard method to implement an operator of the form
$ \sum_{i=0}^{K-1} c_i U_i,$
where $\{U_i\}_{i=0}^{K-1}$ are unitary operators \cite{ChildsWiebe2012}. 
Let $c_i =  r_i e^{\mathrm{i}\theta_i} \in \mathbb{C}, r_i > 0, \theta_i \in [0,2\pi)$. 
We define
$\sqrt{c_i} := \sqrt{r_i} \mathrm{e}^{\mathrm{i}\theta_i/2}.$
Suppose two state-preparation unitaries $V$ and $\widetilde{V}$ such that their first column and first row, respectively, encode the amplitudes $\{\sqrt{c_i}\}_{i=0}^{K-1}$:
\begin{equation}
    V=\frac{1}{\sqrt{\|c\|_1}}
    \begin{pmatrix}
        \sqrt{c_0} & * & \cdots & * \\
        \vdots & \vdots & \ddots & \vdots \\
        \sqrt{c_{K-1}} & * & \cdots & *
    \end{pmatrix},
\end{equation}
\begin{equation}\label{matrix:tildeV}
    \widetilde{V}=\frac{1}{\sqrt{\|c\|_1}}
    \begin{pmatrix}
        \sqrt{c_0} & \cdots & \sqrt{c_{K-1}} \\
        * & \cdots & * \\
        \vdots & \ddots & \vdots \\
        * & \cdots & *
    \end{pmatrix},
\end{equation}
where $\|c\|_1 := \sum_{i=0}^{K-1} |c_i|$.
We also define the select oracle
\begin{equation}
    U := \sum_{i=0}^{K-1} |i\rangle\langle i| \otimes U_i .
\end{equation}
The following lemma shows that these ingredients allow us to implement the desired linear combination.

\begin{lemma}[{LCU lemma~\cite[Lemma 52]{GilyenSuLowEtAl2018}}]\label{lemma:lcu}
Let $c=(c_0,c_1,\ldots,c_{K-1}) \in \mathbb{C}^K$, and define
\(W := (\widetilde{V} \otimes I_m)\, U\, (V \otimes I_m).\)
Then for any input state $|\psi\rangle$,
\[
W\, |0^{\log K}\rangle |\psi\rangle
= \frac{1}{\|c\|_1}\, |0^{\log K}\rangle \sum_{i=0}^{K-1} c_i U_i |\psi\rangle
+ |\widetilde{\perp}\rangle,
\]
where $|\widetilde{\perp}\rangle$ is an unnormalized state satisfying
\[
\bigl(|0^{\log K}\rangle\langle 0^{\log K}| \otimes I_m \bigr)
|\widetilde{\perp}\rangle = 0.
\]
\end{lemma}

Equivalently, the unitary $W$ constitutes a
$(\|c\|_1, \log K, 0)$-block-encoding
of the operator $\sum_{i=0}^{K-1} c_i U_i$.

\subsection{Single-ancilla LCU}
\label{Single-Ancilla LCU}

The conventional LCU framework typically requires a large number of ancilla qubits as well as highly nontrivial multi-qubit controlled unitary operations, which pose significant challenges for near-term and intermediate-scale quantum devices. A recent development \cite{Chakraborty2024} introduces a \emph{single-ancilla} LCU scheme that reduces the ancilla requirement to a single qubit and replaces complicated multi-controlled operations with simpler controlled unitaries, at the expense of an increased number of circuit repetitions.

Let
$S = \sum_{i=0}^{K-1} c_i U_i$
be a linear combination of unitaries and let $|\psi\rangle$ denote the input state. In contrast to the standard LCU construction, which aims to coherently prepare the normalized state $S|\psi\rangle / \|S|\psi\rangle\|$, the single-ancilla LCU method is designed to estimate expectation values of the form
$\langle \psi | S^\dagger O S | \psi \rangle$
for a given observable $O$.

The single-ancilla LCU algorithm assumes that the coefficients $\{c_i\}_{i=0}^{K-1}$ are positive real numbers. This assumption is made without loss of generality. Indeed, for complex coefficients one may write
$c_i = r_i e^{\mathrm{i}\theta_i}, r_i > 0,\theta_i \in [0,2\pi),$
so that
$S = \sum_{i=0}^{K-1} r_i \left( e^{\mathrm{i}\theta_i} U_i \right).$
Since the phase $e^{\mathrm{i}\theta_i}$ can be absorbed into the implementation of the unitary $U_i$, we may assume throughout that all $c_i$ are positive real numbers.

To estimate the expectation value, the algorithm proceeds by randomly selecting unitaries $U_i$ according to the probability distribution
$\left\{ \frac{c_0}{\|c\|_1}, \frac{c_1}{\|c\|_1}, \ldots, \frac{c_{K-1}}{\|c\|_1} \right\},\|c\|_1 := \sum_{i=0}^{K-1} c_i,$
applying them to the input state, and measuring the observable $O$. The detailed procedure is summarized in Algorithm~\ref{algorithm:random_LCU}. The output $\mu$ serves as an estimator of $\langle \psi | S^\dagger O S | \psi \rangle$. 

\begin{algorithm}
\caption{Single-ancilla LCU for expectation observable $\left(O,U_k, \ket{\psi}, T\right)$}

1. Prepare the state $\ket{\psi_1}=\ket{+}\otimes\ket{\psi}$.

2. Sample two independent unitaries $V_1, V_2$ from the distribution $\left\{U_k, \frac{c_k}{\|c\|_1}\right\}$.

3. Define the controlled unitaries 
\[
\widetilde{V}_1 = |0\rangle\langle 0| \otimes I + |1\rangle\langle 1| \otimes V_1,
\qquad
\widetilde{V}_2 = |0\rangle\langle 0| \otimes V_2 + |1\rangle\langle 1| \otimes I,
\] and measure the observable $(X \otimes O)$ on the state
$$
\ket{\psi^\prime}=\tilde{V}_2 \tilde{V}_1 \ket{\psi_1}
$$

4. For the $j^{\text {th }}$ iteration, record the measurement outcome as  $\mu_j$.

5. Repeat Steps 1 to 4 for a total of $T$ times.

6. Output the estimator
$$
\mu=\frac{\|c\|_1^2}{T} \sum_{j=1}^T \mu_j .
$$
\label{algorithm:random_LCU}
\end{algorithm}

\subsection{Compression gadget}
The compression gadget provides a space-efficient way to coherently multiply a sequence of block-encoded operators while reusing the same block-encoding ancilla register. Instead of keeping a separate ancilla register for every factor, it introduces a logarithmic-size counter register that coherently records how many block encodings have been successfully applied~\cite[Lemma~3]{FangLinTong2023}.

\begin{lemma}[Compression gadget]\label{lemma:compression_gadget}
Let $V_1,\ldots,V_L$ be unitaries such that $V_\ell$ is an $(\alpha_\ell,a,0)$-block-encoding of $A_\ell$, then one can construct an $\left(\prod_{\ell=1}^L\alpha_\ell,\,
a+\left\lceil\log_2 L\right\rceil+1,\,0\right)$-block-encoding of the ordered product $A_L\cdots A_2A_1$, using one application of each $V_\ell$.
\end{lemma}

The idea of \cref{lemma:compression_gadget} is illustrated in \cref{fig:compression_gadget}. A counter register of $r=\lceil\log_2 L\rceil+1$ qubits is initialized to $|L\rangle$ by $\mathrm{Add}^L$, where $\mathrm{Add}|c\rangle=|c+1\;\mathrm{mod}\;2^r\rangle$. After each $V_\ell$, a zero-controlled $\mathrm{Add}^{\dagger}$ subtracts one from the counter exactly when the block-encoding ancilla is in $|0^{a}\rangle$. Hence the counter returns to $|0^r\rangle$ if and only if all $L$ block encodings succeed, and postselecting the counter and block-encoding ancillas on zero yields $A_L\cdots A_1/(\prod_{\ell=1}^L\alpha_\ell)$.

\begin{figure}[t]
    \centerline{
    \Qcircuit @C=1em @R=1em {
    \lstick{\ket{0^r}_{\mathrm{ctr}}} & \gate{\mathrm{Add}^{L}} & \gate{\mathrm{Add}^{\dagger}} & \qw & \gate{\mathrm{Add}^{\dagger}} & \qw & \cdots & & \qw & \gate{\mathrm{Add}^{\dagger}} & \qw \\
    \lstick{\ket{0^{a}}_{\mathrm{anc}}} & \multigate{1}{V_1} & \ctrlo{-1} & \multigate{1}{V_2} & \ctrlo{-1} & \qw & \cdots & & \multigate{1}{V_L} & \ctrlo{-1} & \qw \\
    \lstick{\ket{\psi}} & \ghost{V_1} & \qw & \ghost{V_2} & \qw & \qw & \cdots & & \ghost{V_L} & \qw & \qw
    }
    }
    \caption{Compression gadget for coherently implementing the ordered product $A_L\cdots A_1$. The open controls indicate conditioning on the block-encoding ancilla register being in $|0^{a}\rangle$.}
    \label{fig:compression_gadget}
\end{figure}
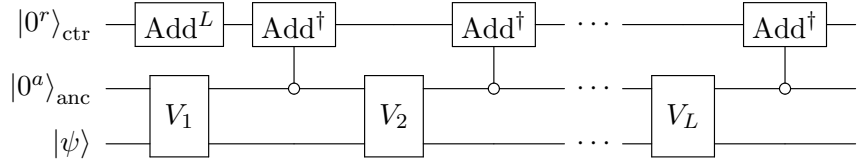

\subsection{Quantum singular value transformation}
\label{subsec:QSVT}

Quantum Singular Value Transformation (QSVT)~\cite{GilyenSuLowEtAl2018} is a powerful framework for implementing polynomial transformations of the singular values of a matrix within a quantum circuit. Depending on the parity of the target polynomial, QSVT admits different constructions. Since our application focuses on matrix inversion, we restrict our attention to the case of \emph{odd} polynomial transformations.

Let $A \in \mathbb{C}^{m \times m}$ be a square matrix. For simplicity, we assume $m = 2^d$ for some positive integer $d$. The singular value decomposition (SVD) of the normalized matrix $A$ can be written as
\begin{equation}
    A = W \Sigma V^{\dagger},
\end{equation}
where $W$ and $V$ are unitary matrices and
\[
\Sigma = \operatorname{diag}(\sigma_1, \sigma_2, \ldots, \sigma_m)
\]
contains the singular values of $A$.

\begin{definition}[Generalized matrix function]\label{definition:SVT_function}
Let $f : \mathbb{R} \to \mathbb{R}$ be a scalar function such that $f(\sigma_i)$ is well defined for all $i = 1,2,\ldots,m$. The generalized matrix function induced by $f$ is defined as
\[
f^{\diamond}(A) := W f(\Sigma) V^{\dagger},
\]
where
\[
f(\Sigma) := \operatorname{diag}\bigl(f(\sigma_1), f(\sigma_2), \ldots, f(\sigma_m)\bigr).
\]
\end{definition}

Suppose that there exists an $(\alpha, a, 0)$-block-encoding of $A$, denoted by $U_A$, such that the singular values of $A/\alpha$ lie in the interval $[0,1]$. Let $p$ be a real-coefficient odd polynomial of degree $n$ satisfying
\[
p(x) \in [-1,1] \qquad \text{for all } x \in [-1,1].
\]
Then QSVT provides an efficient quantum algorithm to implement the operator $p^{\diamond}(A/\alpha)$.
More precisely, $p^{\diamond}(A/\alpha)$ can be realized using a quantum circuit that employs a single ancilla qubit and makes $n$ queries to the block-encoding unitary $U_A$. This construction forms the basis for several quantum linear algebra algorithms, including quantum matrix inversion~\cite{ChakrabortyGilyenJeffery2018}.

\section{Contour integral}\label{sec:contour}

In the introduction, we presented the multivariable matrix eigenvalue transformation problem in \cref{definition:Multi-matrix_function}.
\cite[Theorem 2.4]{HartmannLesch2024} shows that the function satisfying \cref{definition:Multi-matrix_function} is unique, and it can be described via the multi-dimensional contour integral formula as
\begin{equation}\label{eqn:contour_integral_formula}
f(\boldsymbol{A})=f\left(A_{1}, \cdots, A_{n}\right):=\frac{1}{(2 \pi \mathrm{i})^n}\int_{\Gamma_{1}} \cdots \int_{\Gamma_{n}} f(\boldsymbol{z}) \cdot \prod_{j=1}^{n}\left(z_{j}I_m-A_{j}\right)^{-1} \mathrm{~d} \boldsymbol{z},
\end{equation}
where $I_m$ denotes the $m\times m$ identity matrix, $\mathrm{d} \boldsymbol{z}=\mathrm{d} z_{1} \cdots \mathrm{~d} z_{n}$, and $\Gamma_j, 1 \leq \mathrm{j} \leq \mathrm{n}$, are integration contours that encircle $\operatorname{spec}_{\mathcal{A}}\left(A_j\right)$ with $\prod_{j=1}^n\Gamma_j\subset\Omega$.

Assume that for any $j$, $\Gamma_j$ is a piecewise $C^{1,1}$ simple closed curve covering all eigenvalues of $A_j$. 
Let $l_j$ be the length of $\Gamma_j$, $z = z_j(t)$ be its parametrization with $|z_j'(t)| = 1$ in every smooth piece for $t\in[0, l_j]$, and denote $z_j'(t)=e^{i\theta_j(t)}$. 
Let the rate of change of the tangent vector be bounded with $|e^{\mathrm i\theta_j(t)}-e^{\mathrm i\theta_j(t')}|
\le K_{\Gamma_j} |t-t'|$, and $f(z)$ is holomorphic in $\Omega$ and $L$-Lipschitz continuous in $\overline{\Omega}$. 
Then, by change of variable, we have 
\begin{equation}
\begin{aligned}
    f(\boldsymbol{A})
    =&\frac{1}{(2 \pi \mathrm{i})^n}\int_{\Gamma_{1}} \cdots \int_{\Gamma_{n}} f(\boldsymbol{z}) \cdot \prod_{j=1}^{n}\left(z_{j}I_m-A_{j}\right)^{-1} \mathrm{~d} \boldsymbol{z}\\
    =&\frac{1}{(2 \pi \mathrm{i})^n} \int_{0}^{l_1}\cdots \int_{0}^{l_n}f\left(z_1(t_1),\cdots,z_n(t_n)\right)\prod_{j=1}^{n} \left[\left(z_j(t_j) I_m-A_j\right)^{-1} \mathrm{e}^{\mathrm{i} \theta_j(t_j)}\right]\mathrm{~d} t_1\cdots \mathrm{~d} t_n.
\end{aligned}
\label{function:f(A)}
\end{equation}

Next, we approximate $f(\boldsymbol{A})$ by its Riemann sum $f_{\boldsymbol{M}}(\boldsymbol{A})$ defined as 
\begin{equation}
f_{\boldsymbol{M}}(\boldsymbol{A}) = \frac{1}{(2 \pi \mathrm{i})^n}\sum_{k_1=0}^{M_1-1}\cdots\sum_{k_n=0}^{M_n-1} f(z_{1,k_1},\cdots,z_{n,k_n})\prod_{j=1}^{n}\left[\frac{l_j}{M_j} \left(z_{j,k_j} I_m-A_j\right)^{-1} \mathrm{e}^{\mathrm{i} \theta_{j,k_j}}\right], 
\label{function:f_M(A)}
\end{equation}
where $\boldsymbol{M}=(M_1,\cdots,M_n)$, $M_j$ is the number of nodes in the discretization of $\Gamma_j$, and $t_{j,k_j} = \frac{l_jk_j}{M_j}$, $z_{j,k_j} = z_j(t_{j,k_j})$ and $\theta_{j,k_j} =\theta_j(t_{j,k_j})$ for $k_j = 0,1,\cdots, M_{j}-1$. 
The error of $f_{\boldsymbol{M}}(\boldsymbol{A})$ approaching $f(\boldsymbol{A})$ can be bounded by the following proposition.

\begin{proposition}[Discrete integration error by contour integral]
\label{proposition:Discrete_integration_error_by_contour_integral}
Suppose $\Gamma_j$ is a piecewise $C^{1,1}$ simple closed curve covering all eigenvalues of $A_j$ and $K_{\Gamma_j}$ is an upper bound on the tangent variation rate of $\Gamma_j$, $f:\Omega\subset\mathbb{C}^n\to\mathbb{C}$ is holomorphic in $\Omega$ and $L$-Lipschitz continuous on $\overline{\Omega}$, and $B=\max_{\boldsymbol z\in\prod_{j=1}^n\Gamma_j}|f(\boldsymbol z)|$. 
Assume that for every $j$ and every $t_j\in[0,l_j]$ there holds $\|(z_j(t_j)I_m-A_j)^{-1}\|\le\gamma_j$.
Then, for the matrix function $f(A)$ in \cref{function:f(A)} and the approximation $f_{\boldsymbol M}(A)$ in \cref{function:f_M(A)}, the following error bound holds: 
\begin{equation}\label{eq:discrete_error}
\big\|f(\boldsymbol A)-f_{\boldsymbol M}(\boldsymbol A)\big\|
\leq
\frac{1}{(2\pi)^n}\left(\prod_{\ell=1}^nl_\ell\gamma_\ell\right)\left(\sum_{j=1}^n\frac{l_j(B\gamma_j+BK_{\Gamma_j}+L)}{M_j}\right).
\end{equation}

\end{proposition}
\begin{proof}

For piecewise $C^{1,1}$ curves, we can discretize the integration for each segment and sum the results, achieving the same convergence properties as for globally $C^{1,1}$ curves; therefore, without loss of generality, we assume the curve is $C^{1,1}$ globally.
For two values $t_j,t_j'$ we have
\begin{equation}
\left\|(z_j(t_j) I_m -A_j)^{-1} - (z_j(t_j') I_m -A_j)^{-1}\right\| = \left\|\frac{(z_j(t_j)-z_j(t_j'))I}{(z_j(t_j) I_m -A_j)(z_j(t_j') I_m -A_j)}\right\|\leq \gamma_j^2|t_j-t_j'|.
\end{equation}
Suppose $h_j(t) = \frac{1}{2 \pi \mathrm{i}} f(\cdots,z_j(t),\cdots)\left(z_j(t) I_m-A_j\right)^{-1} \mathrm{e}^{\mathrm{i} \theta_j(t)}$, here all variables of $f$ except $z_j(t)$ are fixed.
Because $f(z)$ is $L$-Lipschitz continuous on $\overline{\Omega}$, we have
\begin{equation}
\begin{aligned}
    \left\|h_j(t_j)-h_j(t_j')\right\| \leq &\left\|\frac{1}{2 \pi \mathrm{i}} f(\cdots,z_j(t_j),\cdots)\mathrm{e}^{\mathrm{i} \theta_j(t_j)}\right\| \times \left\|(z_j(t_j) I_m-A_j)^{-1}- (z_j(t_j') I_m-A_j)^{-1}\right\|+ \\ 
    &\left\|\frac{1}{2 \pi \mathrm{i}} f(\cdots,z_j(t_j),\cdots)-\frac{1}{2 \pi \mathrm{i}} f(\cdots,z_j(t_j'),\cdots)\right\| \times \left\|\mathrm{e}^{\mathrm{i} \theta_j(t_j)} (z_j(t_j') I_m-A_j)^{-1}\right\|+ \\ 
    &\left\|\mathrm{e}^{\mathrm{i} \theta_j(t_j)}-\mathrm{e}^{\mathrm{i} \theta_j(t_j')}\right\| \times \left\|\frac{1}{2 \pi \mathrm{i}} f(\cdots,z_j(t_j'),\cdots) (z_j(t_j') I_m-A_j)^{-1}\right\| \\ 
    \leq & \frac{B\gamma_j^2+BK_{\Gamma_j}\gamma_j+L\gamma_j}{2\pi}|t_j-t_j'|.
\end{aligned}
\end{equation}

Let
\begin{equation}
H(t_1,\dots,t_n):=\frac{1}{(2\pi\mathrm{i})^n}
\,f\left(z_1(t_1),\dots,z_n(t_n)\right)\prod_{j=1}^n\left[(z_j(t_j)I_m-A_j)^{-1}\mathrm{e}^{\mathrm{i}\theta_j(t_j)}\right].
\end{equation}

Then $f(\boldsymbol A)=\int_{[0,l_1]\times\cdots\times[0,l_n]}H(\boldsymbol t)\,\mathrm{d}\boldsymbol t$.
Choose one index $j\in\{1,\dots,n\}$ and keep the other coordinates $t_\ell$ ($\ell\neq j$) fixed, then 
\begin{equation}
\begin{aligned}
    &\|H(\dots,t_j,\dots)-H(\dots,t_j',\dots)\|\\
    \leq&\left\|\frac{1}{(2\pi\mathrm{i})^{n-1}}\prod_{\ell\neq j}\left[(z_\ell(t_\ell)I_m-A_\ell)^{-1} \mathrm{e}^{\mathrm{i}\theta_\ell(t_\ell)}\right]\right\|\cdot\left\|h_j(t_j)-h_j(t_j')\right\|\\
    \leq&\frac{(B\gamma_j+BK_{\Gamma_j}+L)\prod_{\ell=1}^n\gamma_\ell}{(2\pi)^n}|t_j-t_j'|.
\end{aligned}
\end{equation}

Therefore, for $(t_1,\cdots,t_n)\in[t_{1,k_1},t_{1,k_{1}+1}]\times\cdots\times[t_{n,k_n},t_{n,k_{n}+1}]$ we have $|t_j-t_{j,k_j}|\leq\frac{l_j}{M_j}$ and
\begin{equation}
\begin{aligned}
&\left\|\int_{[t_{1,k_1},t_{1,k_{1}+1}]\times\cdots\times[t_{n,k_n},t_{n,k_{n}+1}]}H(t_1,\cdots,t_n)\,\mathrm{d}t_1\cdots\mathrm{d}t_n-H(t_{1,k_1},\cdots,t_{n,k_n})\prod_{j=1}^n\frac{l_j}{M_j}\right\|\\
=&\left\|\int_{[t_{1,k_1},t_{1,k_{1}+1}]\times\cdots\times[t_{n,k_n},t_{n,k_{n}+1}]}\left(H(t_1,\cdots,t_n)-H(t_{1,k_1},\cdots,t_{n,k_n})\right)\,\mathrm{d}t_1\cdots\mathrm{d}t_n\right\|\\
\leq&\sum_{j=1}^n\left\|\int_{[t_{1,k_1},t_{1,k_{1}+1}]\times\cdots\times[t_{n,k_n},t_{n,k_{n}+1}]}\left(H(\dots,t_j,\dots)-H(\dots,t_{j,k_j},\dots)\right)\mathrm{d}t_1\cdots\mathrm{d}t_n\right\|\\
\leq&\sum_{j=1}^n\int_{[t_{1,k_1},t_{1,k_{1}+1}]\times\cdots\times[t_{n,k_n},t_{n,k_{n}+1}]}\frac{(B\gamma_j+BK_{\Gamma_j}+L)\prod_{\ell=1}^n\gamma_\ell}{(2\pi)^n}|t_j-t_{j,k_j}|\,\mathrm{d}t_1\cdots\mathrm{d}t_n\\
\leq&\frac{1}{(2\pi)^n}\left(\prod_{\ell=1}^n\frac{l_\ell\gamma_\ell}{M_\ell}\right)\left(\sum_{j=1}^n\frac{l_j(B\gamma_j+BK_{\Gamma_j}+L)}{M_j}\right).
\end{aligned}
\end{equation}

Then summarizing all of these integrals, we can get the result

\begin{equation}
\begin{aligned}
&\big\|f(\boldsymbol A)-f_{\boldsymbol M}(\boldsymbol A)\big\|\\
\leq&\sum_{k_1=0}^{M_1-1}\cdots\sum_{k_n=0}^{M_n-1}\left\|\int_{[t_{1,k_1},t_{1,k_{1}+1}]\times\cdots\times[t_{n,k_n},t_{n,k_{n}+1}]}H(t_1,\cdots,t_n)\,\mathrm{d}t_1\cdots\mathrm{d}t_n-H(t_{1,k_1},\cdots,t_{n,k_n})\prod_{j=1}^n\frac{l_j}{M_j}\right\|\\
\leq&\left(\prod_{r=1}^nM_r\right)\cdot\frac{1}{(2\pi)^n}\left(\prod_{\ell=1}^n\frac{l_\ell\gamma_\ell}{M_\ell}\right)\left(\sum_{j=1}^n\frac{l_j(B\gamma_j+BK_{\Gamma_j}+L)}{M_j}\right) \\
& =\frac{1}{(2\pi)^n}\left(\prod_{\ell=1}^nl_\ell\gamma_\ell\right)\left(\sum_{j=1}^n\frac{l_j(B\gamma_j+BK_{\Gamma_j}+L)}{M_j}\right).
\end{aligned}
\end{equation}

\end{proof}

\subsection{Quantum algorithm}
\label{subsection:Contour Integral Algorithm}
Our task is to approximate $f_{\boldsymbol M}(\boldsymbol A)$ given $\Gamma_1,\cdots,\Gamma_n$, $\boldsymbol M$ and $\boldsymbol A$. We start with the initial state $|\psi\rangle$. 
Note that for $j=1,\cdots,n$ we have block encoding \( U_{A_j} \) of the matrix \( A_j \). 
Then
\begin{equation}
    U_{A_j}\ket{0^{a}}\ket{\psi} = \frac{1}{\alpha_j}\ket{0^{a}}A_j\ket{\psi} +\ket{\perp}.
\end{equation}
For $z_{j,k_j} = z_j(t_{j,k_j})$, $j=1,\cdots,n$, let $U_j$ be select oracles, $V_j$ be the controlled rotations satisfying
\begin{equation}
    V_j\ket{k_j}\ket{0}_{V_j} = \frac{1}{\sqrt{\alpha_j + |z_{j,k_j}|}}\ket{k_j}\left(\sqrt{z_{j,k_j}}\ket{0}+\mathrm{i}\sqrt{\alpha_j}\ket{1}\right),
    \label{operator:control_rotation}
\end{equation}
\begin{equation}
    U_j = \ket{0}\bra{0} \otimes I_m + \ket{1}\bra{1}\otimes U_{A_j}.
\end{equation}
Thus, applying \cref{lemma:lcu}, we have that
\begin{equation}
\label{select_oracle}
    \left(\widetilde{V_j} \otimes I_{a}\otimes I_{m}\right) (I_{k_j}\otimes U_j)\left(V_j \otimes I_{a}\otimes I_{m}\right)\ket{k_j} \ket{0}_{V_j}\ket{0^{a}}\ket{\psi} = \frac{\ket{k_j}\ket{0}_{V_j}\ket{0^{a}}(z_{j,k_j} I_{m}-A_j)\ket{\psi}}{\alpha_j + |z_{j,k_j}|}+\ket{\widetilde{\perp}}, 
\end{equation}
where $\widetilde{V_j}$ is a control rotation, and its action on the last qubit is constructed from \cref{matrix:tildeV}.
Among them, $\ket{\widetilde{\perp}}$ refers to both the vertical component in block encoding and LCU. Thus we have that \cref{select_oracle} is an $(\alpha_j + |z_{j,k_j}|, a+1, 0)$-block-encoding of $z_{j,k_j} I_m-A_j$ controlled by $\ket{k_j}$, and we simply call it $U_{j,k_j}$.  

As $x^{-1}$ is an odd function on $[-1,1]$ and singular at $x = 0$, we can use an odd polynomial $p_j(x)$ to approximate $\frac{3\delta_j}{4x}$ on $[-1,-\delta_j] \cup[\delta_j, 1]$, and $p_j^\diamond\left(\frac{z_{j,k_j} I_m-A_j}{\alpha_j + |z_{j,k_j}|}\right)^\dagger$ can be implemented by QSVT~\cite{GilyenSuLowEtAl2018}, where $p_j^\diamond(U)$ is defined by \cref{definition:SVT_function}. The control circuit of  \cref{select_oracle} can be introduced over the QSVT circuit like \cref{fig:circuit_Control_QSVT}, which gives a block encoding of $(z_{j,k_j} I_m-A_j)^{-1}$, and we call this circuit $\widetilde{U}_k$.

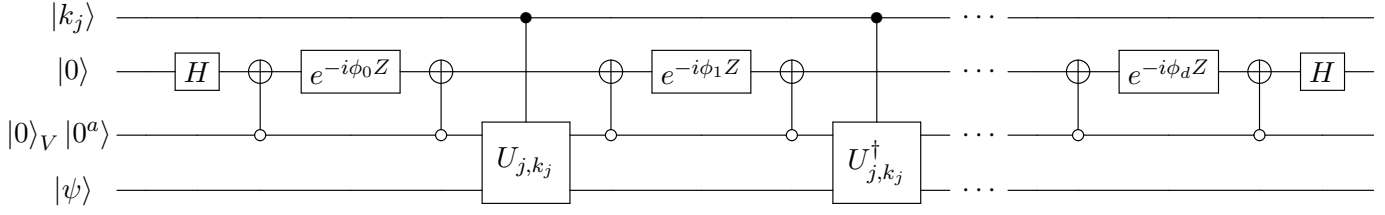
\begin{figure}
    \centerline{
    \Qcircuit @R=1em @C=1em {
    \ket{k_j} \quad\quad\quad & 
    \qw & \qw & \qw & \qw & \qw & 
    \ctrl{2} & \qw & \qw & \qw & \ctrl{2} &
    \qw & \cdots & & \qw & \qw & \qw & \qw & \qw & \qw &\\
    \ket{0} \quad\quad\quad &
    \qw & \gate{H} &\targ & \gate{e^{-i \phi_0 Z}} & \targ &
    \qw & \targ & \gate{e^{-i \phi_{1} Z}} & \targ & \qw & 
    \qw & \cdots & & \qw  &\targ & \gate{e^{-i \phi_d Z}} & \targ & \gate{H} & \qw & \\ 
    \ket{0}_V\ket{0^a} \quad\quad\quad\quad & 
    \qw & \qw & \ctrlo{-1} & \qw & \ctrlo{-1} & 
    \multigate{1}{U_{j,k_j}} & \ctrlo{-1} & \qw & \ctrlo{-1} & \multigate{1}{U_{j,k_j}^\dagger} & 
    \qw & \cdots & & \qw & \ctrlo{-1} & \qw & \ctrlo{-1} & \qw & \qw &  \\
    \ket{\psi} \quad\quad\quad & 
    \qw & \qw & \qw & \qw & \qw & 
    \ghost{U_{j,k_j}} & \qw & \qw & \qw & \ghost{U_{j,k_j}^\dagger} & 
    \qw & \cdots & & \qw & \qw & \qw & \qw & \qw & \qw &  \\
    }
    }
    \caption{QSVT circuit with control gate $\ket{k}$. This is a block encoding of $(z_{j,k_j} I_m-A_j)^{-1}$, and it's the sub-circuit that implements $f(\boldsymbol A)$. We call this entire circuit $\widetilde{U}_{k_j}$, which will be used in subsequent line.}
    \label{fig:circuit_Control_QSVT}
\end{figure} 

Through compression gadget \cref{lemma:compression_gadget}, we can generate a block encoding of $\prod_{j=1}^{n} \left(z_{j,k_j} I_m-A_j\right)^{-1}$ by calling the circuit showed in \cref{fig:circuit_Control_QSVT}. The block encoding of $\prod_{j=1}^{n} \left(z_{j,k_j} I_m-A_j\right)^{-1}$ is shown \cref{fig:circuit_multi_ctrl}, and due to the control gate, it can be used directly as the select oracle $\mathrm{SEL}$ in LCU.
The prepare oracle in the corresponding LCU is a unitary such that 
\begin{equation}
    C\ket{0}\cdots\ket{0}=\frac{1}{\sqrt{\|\boldsymbol c\|_1}} \sum_{k_1=0}^{M_1-1}\cdots\sum_{k_n=0}^{M_n-1} \sqrt{c_{\boldsymbol k}}\ket{k_1}\cdots\ket{k_n},
\end{equation}
where $\boldsymbol c=(c_{(0,\cdots,0)},c_{(0,\cdots,1)},\cdots,c_{(M_1-1,\cdots,M_n-1)})$, and
\begin{equation}
c_{\boldsymbol k}=f\left(z_{1,k_1},\dots,z_{n,k_n}\right) \prod_{j=1}^n\frac{2}{3\delta_j\pi\mathrm{i}}\cdot \frac{\mathrm{e}^{\mathrm{i} \theta_{j,k_j}}l_j}{(\alpha_j + |z_{j,k_j}|)M_j}.
\end{equation}

Notice that 
\begin{equation}
\begin{aligned}
&\sum_{k_1=0}^{M_1-1}\cdots\sum_{k_n=0}^{M_n-1} c_{\boldsymbol k}\prod_{j=1}^{n} p_j^\diamond\left(\frac{z_{j,k_j} I_m-A_j}{\alpha_j + |z_{j,k_j}|}\right)^\dagger\\
\approx&\sum_{k_1=0}^{M_1-1}\cdots\sum_{k_n=0}^{M_n-1} f\left(z_{1,k_1},\dots,z_{n,k_n}\right) \prod_{j=1}^n\frac{2}{3\delta_j\pi\mathrm{i}}\cdot \frac{\mathrm{e}^{\mathrm{i}\theta_{j,k_j}}l_j}{(\alpha_j + |z_{j,k_j}|)M_j} \cdot \frac{3\delta_j(\alpha_j + |z_{j,k_j}|)}{4}(z_{j,k_j} I_m-A_j)^{-1}\\
=& \frac{1}{(2 \pi \mathrm{i})^n}\sum_{k_1=0}^{M_1-1}\cdots\sum_{k_n=0}^{M_n-1} f(z_{1,k_1},\cdots,z_{n,k_n})\prod_{j=1}^{n}\left[\frac{l_j}{M_j} \left(z_{j,k_j} I_m-A_j\right)^{-1} \mathrm{e}^{\mathrm{i} \theta_{j,k_j}}\right]. 
\end{aligned}
\end{equation}
By \cref{lemma:lcu}, the operator 
\begin{equation}\label{operator:total_circuit}
     \left(\widetilde{C} \otimes I \right)\mathrm{SEL}\left(C\otimes I\right)
\end{equation}
gives a block encoding of $$\sum_{k_1=0}^{M_1-1}\cdots\sum_{k_n=0}^{M_n-1} c_{\boldsymbol k}\prod_{j=1}^{n} p_j^\diamond\left(\frac{z_{j,k_j} I_m-A_j}{\alpha_j + |z_{j,k_j}|}\right)^\dagger, $$ which is an approximation of $f_{\boldsymbol M}(\boldsymbol A)$.

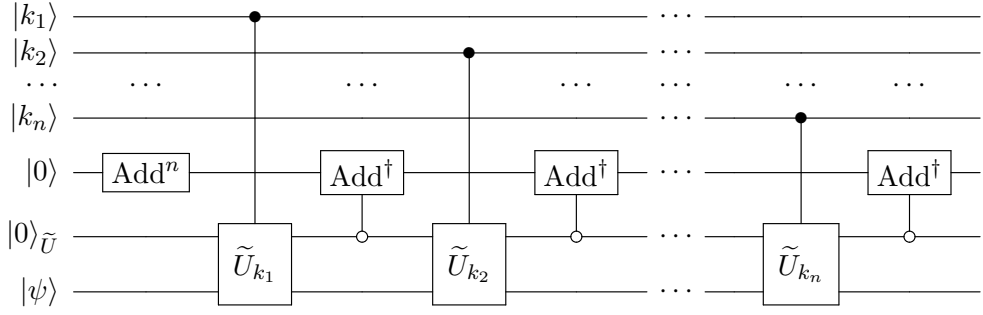
\begin{figure}
    \centerline{
\Qcircuit @C=1em @R=1em {
\lstick{\ket{k_1}}    & \qw         & \ctrl{5} & \qw      & \qw      & \qw        & \qw & \cdots&& \qw& \qw& \qw& \qw\\
\lstick{\ket{k_2}}    & \qw         & \qw      & \qw      & \ctrl{4} & \qw      & \qw & \cdots  && \qw& \qw& \qw& \qw\\
\lstick{\cdots}       &    \cdots   &          & \cdots   &       & \cdots   &    & \cdots   & &      & \cdots  & \cdots & \\
\lstick{\ket{k_n}}    & \qw         & \qw      &  \qw     & \qw & \qw      & \qw & \cdots && \qw& \ctrl{2}& \qw& \qw\\
\lstick{\ket{0}}      & \gate{\mathrm{Add}^n}& \qw      & \gate{\mathrm{Add}^\dagger} & \qw & \gate{\mathrm{Add}^\dagger} & \qw & \cdots && \qw & \qw & \gate{\mathrm{Add}^\dagger} & \qw\\
\lstick{\ket{0}_{\widetilde{U}}}& \qw         &\multigate{1}{\widetilde{U}_{k_1}}& \ctrlo{-1} & \multigate{1}{\widetilde{U}_{k_2}} & \ctrlo{-1} & \qw & \cdots && \qw &\multigate{1}{\widetilde{U}_{k_n}}& \ctrlo{-1}& \qw\\
\lstick{\ket{\psi}}   & \qw         & \ghost{\widetilde{U}_{k_1}}      & \qw      & \ghost{\widetilde{U}_{k_2}} & \qw      & \qw & \cdots && \qw &\ghost{\widetilde{U}_{k_n}}& \qw& \qw\\
}
}
    \caption{Controlled block encoding of $\prod_{j=1}^{n} \left(z_{j,k_j} I_m-A_j\right)^{-1}$, which also serves as the select oracle $\mathrm{SEL}$.}
    \label{fig:circuit_multi_ctrl}
\end{figure}

\subsection{Complexity analysis}

In the previous discussion, we have showed that we use $p_j(x)$ to approximate $\frac{3\delta_j}{4x}$ on $[-1,-\delta_j] \cup[\delta_j, 1]$. Suppose that
\begin{equation}
    \left|p_j(x)-\frac{3 \delta_j}{4 x}\right| \leq \epsilon^{\prime}, \quad \forall x \in[-1,-\delta_j] \cup[\delta_j, 1]
\end{equation} 
and $|p_j(x)| \leq 1$ for all $x \in [-1,1]$. The existence of such an odd polynomial of degree $\mathcal{O}(\frac{1}{\delta_j}\mathrm{log}(\frac{1}{\epsilon^\prime}))$ is guaranteed by \cite[Corollary 69]{GilyenSuLowEtAl2018}. 

Assume that there is singular value decomposition $$\frac{z_{j,k_j} I_m-A_j}{\alpha_j + |z_{j,k_j}|}=W_{j,k_j} \frac{\Sigma_{j,k_j}}{\alpha_j+|z_{j,k_j}|} V_{j,k_j}^{\dagger},$$ 
then \cite[Theorem 2]{GilyenSuLowEtAl2018}  enables us to implement $$p_j^\diamond\left(\frac{z_{j,k_j} I_m-A_j}{\alpha_j + |z_{j,k_j}|}\right)^\dagger=V_{j,k_j} p_j\left(\frac{\Sigma_{j,k_j}}{\alpha_j+|z_{j,k_j}|}\right) W_{j,k_j}^{\dagger}.$$ 
If all diagonal elements of $\frac{\Sigma_{j,k_j}}{\alpha_j+|z_{j,k_j}|}$, i.e. the singular values of $\frac{z_{j,k_j} I_m-A_j}{\alpha_j + |z_{j,k_j}|}$, are in the interval $[\delta_j , 1]$ for $k_j = 0,1,\cdots,M_j-1$, then we have 
\begin{equation}\label{QSVT_error}
\left\|p_j^{\diamond}\left(\frac{z_{j,k_j} I_m-A_j}{\alpha_j + |z_{j,k_j}|}\right)^{\dagger}-\frac{3 \delta_j}{4}\left(\frac{z_{j,k_j} I_m-A_j}{\alpha_j + |z_{j,k_j}|}\right)^{-1}\right\|=\left\|p_j\left(\frac{\Sigma_{j,k_j}}{\alpha_j+|z_{j,k_j}|}\right)-\frac{3 \delta_j}{4}\left(\frac{\Sigma_{j,k_j}}{\alpha_j+|z_{j,k_j}|}\right)^{-1}\right\| \leq \epsilon^{\prime}.
\end{equation}

Assume $\Gamma_j$ is inside disk $|z| \leq R_j$, then $0\leq\alpha_j + |z_{j,k_j}| \leq \alpha_j+R_j$. Under the assumption that for any $t \in [0, l]$, $\left\|(z_j(t) I_m -A_j)^{-1}\right\| \leq \gamma_j$,  the singular values of $\frac{z_{j,k_j} I_m-A_j}{\alpha_j + |z_{j,k_j}|}$ are all larger than $\delta_j = \frac{1}{\gamma_j(\alpha_j+R_j)}$, 
\begin{equation}
\|\boldsymbol c\|_1 \leq B\prod_{j=1}^n\frac{2l_j}{3\pi\alpha_j\delta_j}= B\prod_{j=1}^n\frac{2\gamma_j l_j(1+R_j/\alpha_j)}{3\pi}.
\end{equation}
Since we use $\sum_{k_1=0}^{M_1-1}\cdots\sum_{k_n=0}^{M_n-1} c_{\boldsymbol k}\prod_{j=1}^{n} p_j^\diamond\left(\frac{z_{j,k_j} I_m-A_j}{\alpha_j + |z_{j,k_j}|}\right)^\dagger$ to approximate $f_{\boldsymbol M}(\boldsymbol A)$, we have
\begin{equation}
\label{equation:QSVT_error}
\begin{aligned}
    &\left\|\sum_{k_1=0}^{M_1-1}\cdots\sum_{k_n=0}^{M_n-1} c_{\boldsymbol k}\prod_{j=1}^{n} p_j^\diamond\left(\frac{z_{j,k_j} I_m-A_j}{\alpha_j + |z_{j,k_j}|}\right)^\dagger -f_{\boldsymbol M}(\boldsymbol A)\right\| \\
    \leq&\left(\sum_{k_1=0}^{M_1-1}\cdots \sum_{k_n=0}^{M_n-1} |c_{\boldsymbol k}|\right)\left\| \prod_{j=1}^{n} p_j^\diamond\left(\frac{z_{j,k_j} I_m-A_j}{\alpha_j + |z_{j,k_j}|}\right)^\dagger -\prod_{j=1}^{n}(3 \delta_j / 4)\left(\frac{z_{j,k_j} I_m-A_j}{\alpha_j + |z_{j,k_j}|}\right)^{-1}\right\|\leq \|\boldsymbol c\|_1n\epsilon^{\prime}.
\end{aligned}
\end{equation}
When the algorithm succeeds, we get the state $$\sum_{k_1=0}^{M_1-1}\cdots\sum_{k_n=0}^{M_n-1} c_{\boldsymbol k}\prod_{j=1}^{n} p_j^\diamond\left(\frac{z_{j,k_j} I_m-A_j}{\alpha_j + |z_{j,k_j}|}\right)^\dagger\ket{\psi}\left/\left\|\sum_{k_1=0}^{M_1-1}\cdots\sum_{k_n=0}^{M_n-1} c_{\boldsymbol k}\prod_{j=1}^{n} p_j^\diamond\left(\frac{z_{j,k_j} I_m-A_j}{\alpha_j + |z_{j,k_j}|}\right)^\dagger\ket{\psi}\right\|.\right.$$ 
Also, choose a proper $\boldsymbol M$ so that  
\begin{equation}\label{condition:M}
    \left\|f(\boldsymbol A)\ket{\psi}-f_{\boldsymbol M}(\boldsymbol A)\ket{\psi}\right\| \leq \|\boldsymbol c\|_1 n\epsilon^\prime,
\end{equation}
then combine these two equations together, we have
\begin{equation}
    \left\|f(\boldsymbol A)\ket{\psi}-\sum_{k_1=0}^{M_1-1}\cdots\sum_{k_n=0}^{M_n-1} c_{\boldsymbol k}\prod_{j=1}^{n} p_j^\diamond\left(\frac{z_{j,k_j} I_m-A_j}{\alpha_j + |z_{j,k_j}|}\right)^\dagger\ket{\psi}\right\| \leq 2\|\boldsymbol c\|_1n\epsilon^\prime,
\end{equation}
and due to that $\|\frac{u}{\|u\|}-\frac{v}{\|v\|}\|\leq2\frac{\|u-v\|}{\|u\|}$, we have
\begin{equation}
    \left\|\frac{f(\boldsymbol A)\ket{\psi}}{\|f(\boldsymbol A)\ket{\psi}\|}-\frac{\sum_{k_1=0}^{M_1-1}\cdots\sum_{k_n=0}^{M_n-1} c_{\boldsymbol k}\prod_{j=1}^{n} p_j^\diamond\left(\frac{z_{j,k_j} I_m-A_j}{\alpha_j + |z_{j,k_j}|}\right)^\dagger\ket{\psi}}{\left\|\sum_{k_1=0}^{M_1-1}\cdots\sum_{k_n=0}^{M_n-1} c_{\boldsymbol k}\prod_{j=1}^{n} p_j^\diamond\left(\frac{z_{j,k_j} I_m-A_j}{\alpha_j + |z_{j,k_j}|}\right)^\dagger\ket{\psi}\right\|}\right\| \leq \frac{4\|\boldsymbol c\|_1n\epsilon^\prime}{\|f(\boldsymbol A)\ket{\psi}\|}.
\end{equation}
Therefore, we can let 
$$
\epsilon^\prime = \frac{\|f(\boldsymbol A)\ket{\psi}\|\epsilon}{4n\|\boldsymbol c\|_1}.
$$
As $\|\boldsymbol c\|_1 \leq B\prod_{j=1}^n\frac{2\gamma_j l_j(1+R_j/\alpha_j)}{3\pi}$, we have $$\frac{1}{\epsilon^\prime}= \mathcal{O}\left(\frac{B\prod_{j=1}^n\frac{2\gamma_j l_j(1+R_j/\alpha_j)}{3\pi}}{n\|f(\boldsymbol A)\ket{\psi}\|\epsilon}\right),$$ 
since the spectral radius of $A$ is no greater than $\|A\|$, the total number of calls to $U_{A_j}$ by $\widetilde{U}_{k_j}$ is 
\begin{equation}
    \mathcal{O} \left(\frac{1}{\delta_j}\mathrm{log}\frac{1}{\epsilon^\prime}\right) = \mathcal{O}\left(\gamma_j(\alpha_j+R_j)\mathrm{log}\frac{B\prod_{j=1}^nl_j\gamma_j}{\|f(\boldsymbol A)\ket{\psi}\|\epsilon}\right).
\end{equation}

As the whole circuit is given by Eq. \eqref{operator:total_circuit}, we should consider $C,\widetilde{C}$ and $\mathrm{SEL}$ respectively when calculating complexity. For any given unit vector $\ket{\psi}$, the success probability of Eq. \eqref{operator:total_circuit} is 
$$
\left\|\sum_{k_1=0}^{M_1-1}\cdots\sum_{k_n=0}^{M_n-1} c_{\boldsymbol k}\prod_{j=1}^{n} p_j^\diamond\left(\frac{z_{j,k_j} I_m-A_j}{\alpha_j + |z_{j,k_j}|}\right)^\dagger\ket{\psi}\right\|^2/\|\boldsymbol c\|_1^2 .  
$$
Therefore, with amplitude amplification, 
$$
\mathcal{O}\left(\|\boldsymbol c\|_1/\left\|\sum_{k_1=0}^{M_1-1}\cdots\sum_{k_n=0}^{M_n-1} c_{\boldsymbol k}\prod_{j=1}^{n} p_j^\diamond\left(\frac{z_{j,k_j} I_m-A_j}{\alpha_j + |z_{j,k_j}|}\right)^\dagger\ket{\psi}\right\|\right)
$$
calls to Eq. \eqref{operator:total_circuit} can make the algorithm succeed with probability at least $1/2$. 

According to \cref{proposition:Discrete_integration_error_by_contour_integral}, in order to suit Eq. \eqref{condition:M} we need to choose $\boldsymbol M$ so that 
\begin{equation}
\frac{1}{(2\pi)^n}\left(\prod_{\ell=1}^nl_\ell\gamma_\ell\right)\left(\sum_{j=1}^n\frac{l_j(B\gamma_j+BK_{\Gamma_j}+L)}{M_j}\right)\leq\epsilon^\prime = \frac{\|f(\boldsymbol A)\ket{\psi}\|\epsilon}{4n\|\boldsymbol c\|_1}, 
\end{equation}
Therefore we can let 
\begin{equation}
M_j\geq nl_j\frac{(B\gamma_j + BK_{\Gamma_j} + L)\cdot(\prod_{\ell=1}^n l_\ell\gamma_\ell)\cdot4n\|\boldsymbol c\|_1}{(2\pi)^n\|f(\boldsymbol A)\ket{\psi}\|\epsilon}
\end{equation}
so that 
\begin{equation}
    \frac{(B\gamma_j + BK_{\Gamma_j} + L)\,\prod_{\ell=1}^n l_\ell\gamma_\ell}{(2\pi)^n}\cdot \frac{l_j}{M_j}\leq\frac{\epsilon^\prime}{n}
\end{equation}
for each $j$. In this situation,
\begin{equation}
\begin{aligned}
\log M_j = &\mathcal{O}\left(\log\left(nl_j\frac{(B\gamma_j + BK_{\Gamma_j} + L)\cdot\prod_{\ell=1}^n l_\ell\gamma_\ell\cdot4n\|\boldsymbol c\|_1}{(2\pi)^n\|f(\boldsymbol A)\ket{\psi}\|\epsilon}\right)\right)\\
=&\mathcal{O}\left(\log\frac{n(B\gamma_j + BK_{\Gamma_j}+L)\prod_{\ell=1}^n l_\ell\gamma_\ell}{\|f(\boldsymbol A)\ket{\psi}\|\epsilon}\right).
\end{aligned}
\end{equation}
So we need a total of $w=\mathcal{O}\left(\sum_{j=1}^n\log\frac{(B\gamma_j + BK_{\Gamma_j}+L)\prod_{\ell=1}^n l_\ell\gamma_\ell}{\|f(\boldsymbol A)\ket{\psi}\|\epsilon}\right)+\log n+a+3$ ancilla qubits and Eq. \eqref{operator:total_circuit} is a $(\|\boldsymbol c\|_1, w, 2\|\boldsymbol c\|_1n\epsilon^\prime)$-block-encoding of $f(\boldsymbol A)$.

Suppose that $\epsilon$ is small enough, then $\sum_{k_1=0}^{M_1-1}\cdots\sum_{k_n=0}^{M_n-1} c_{\boldsymbol k}\prod_{j=1}^{n} p_j^\diamond\left(\frac{z_{j,k_j} I_m-A_j}{\alpha_j + |z_{j,k_j}|}\right)^\dagger\ket{\psi}$ and $f(\boldsymbol A)\ket{\psi}$ have the same order. In other words, $$\|\boldsymbol c\|_1/\left\|\sum_{k_1=0}^{M_1-1}\cdots\sum_{k_n=0}^{M_n-1} c_{\boldsymbol k}\prod_{j=1}^{n} p_j^\diamond\left(\frac{z_{j,k_j} I_m-A_j}{\alpha_j + |z_{j,k_j}|}\right)^\dagger\ket{\psi}\right\|= \mathcal{O}\left(\frac{B\prod_{j=1}^n\frac{2\gamma_j l_j(1+R_j/\alpha_j)}{3\pi}}{\left\|f(\boldsymbol A)\ket{\psi}\right\|}\right).$$ 
Summarizing these calculations, we can get the following result. 

\begin{theorem}[Complexity of the contour integral approach]
Let $U_{A_j}$ be an $(\alpha_j,a,0)$-block-encoding of $A_j$ for $j=1,\cdots,n$. $\Gamma_j$ is a piecewise $C^{1,1}$ simple closed curve covering all eigenvalues of $A_j$ and $K_{\Gamma_j}$ is an upper bound on the tangent variation rate of $\Gamma_j$. Under the assumption that for any $z \in \Gamma_j$, $\left\|(z I_m -A_j)^{-1}\right\| \leq \gamma_j$, we can solve the problem defined by Definition \ref{definition:main_problem} using a quantum circuit with 
$$\mathcal{O}\left(\frac{B\prod_{r=1}^n\frac{2\gamma_r l_r(1+R_r/\alpha_r)}{3\pi}}{\left\|f(\boldsymbol A)\ket{\psi}\right\|}\cdot\gamma_j(\alpha_j+R_j)\mathrm{log}\frac{B \prod_{\ell=1}^nl_\ell\gamma_\ell}{\|f(\boldsymbol A)\ket{\psi}\|\epsilon}\right)$$ queries to $U_{A_j}$, $$\mathcal{O}\left(\frac{B\prod_{r=1}^n\frac{2\gamma_r l_r(1+R_r/\alpha_r)}{3\pi}}{\left\|f(\boldsymbol A)\ket{\psi}\right\|}\right) $$ queries to $U_{\psi}$ and $$\mathcal{O}\left(\sum_{j=1}^n\log\frac{n(B\gamma_j + BK_{\Gamma_j}+L)\prod_{\ell=1}^n l_\ell\gamma_\ell}{\|f(\boldsymbol A)\ket{\psi}\|\epsilon}\right)+\log n+a+3$$ ancilla qubits.

\label{theorem:contour_integral_complexity}
\end{theorem}

In addition to the select oracle, another part of the computation is to construct prepare oracles $C$ and $\widetilde{C}$. Notice that the prepare oracles $C$ and $\widetilde{C}$ prepare superpositions of $\widehat{M}=\prod_{j=1}^nM_j$ basis states. In general, we can prepare an $\widehat{M}$-dimensional quantum state with cost $\mathcal{O}(\widehat{M})$ \cite{SBM06}, but this might incur a gate complexity that is polynomial in the inverse error $\frac{1}{\|f(\boldsymbol A)\ket{\psi}\|\epsilon}$, as $\widehat{M}$ scales polynomially in $\frac{1}{\epsilon}$. However, since the amplitudes of the states are known integrable functions evaluated at discrete points, the state preparation circuits can be constructed more efficiently, in time only $\mathcal{O}(\mathrm{poly log} \widehat{M})$ \cite{GroverRudolph2002,McardleGilyenBerta2022}. 

\subsection{Ancilla-efficient algorithm}
\label{subsection:Ancilla-efficient_Algorithm_contour}
The version in the previous section achieves good results in terms of query complexity, but used many ancilla qubits and a complex control circuit Eq. \eqref{operator:control_rotation}, which could be a huge challenge for early fault-tolerant quantum computers. 
On the other hand, Section \ref{Single-Ancilla LCU} gives us a method to simplify the circuit and reduce the number of ancilla qubits. This method no longer gets an approximation of the target state $f(\boldsymbol A)\ket{\psi} / \| f(\boldsymbol A)\ket{\psi} \|$, but the observation value of the target state under $O$. The question is defined by Definition \ref{definition:estimation_problem}.

In this method, for any $j$ we no longer use complex control gate $V_j$, but for every $k_j$, we implement
\begin{equation}
    V_{j,k_j}\ket{0}_{V_j} = \frac{1}{\sqrt{\alpha_j + |z_{j,k_j}|}}\left(\sqrt{z_{j,k_j}}\ket{0}+ \mathrm{i}\sqrt{\alpha_j}\ket{1}\right),
\end{equation}
then follow the steps above. We can get $U_{j,k_j}$ to be an $(1,a+2)$-block-encoding of $p_j^\diamond\left(\frac{z_{j,k_j} I_m-A_j}{\alpha_j + |z_{j,k_j}|}\right)^\dagger$ for each $k_j$ respectively. Then with compression gadget \cref{lemma:compression_gadget}, we can get a block encoding of $$\prod_{j=1}^{n} p_j^\diamond\left(\frac{z_{j,k_j} I_m-A_j}{\alpha_j + |z_{j,k_j}|}\right)^\dagger$$ with $a+3+\log n$ ancilla qubits, and we name this block encoding $U_{\boldsymbol{k}}$.

\cite{JiangAn2026} shows that we can let $\ket{\widetilde{\psi}}=\ket{0^{a+3+\log n}}\otimes\ket{\psi}$ and $\widetilde{O}=\ket{0^{a+3+\log n}}\bra{0^{a+3+\log n}}\otimes O$, then 
\begin{equation}
\bra{\widetilde{\psi}}U_{\boldsymbol{k}}^\dagger\widetilde{O}U_{\boldsymbol{k}}\ket{\widetilde{\psi}}
=\bra{\psi} \left(\prod_{j=1}^{n} p_j^\diamond\left(\frac{z_{j,k_j} I_m-A_j}{\alpha_j + |z_{j,k_j}|}\right)^\dagger\right)^\dagger O \prod_{\ell=1}^{n} p_\ell^\diamond\left(\frac{z_{\ell,k_\ell} I_m-A_\ell}{\alpha_\ell + |z_{\ell,k_\ell}|}\right)^\dagger\ket{\psi}.
\end{equation}
Therefore, we can use the algorithm in Section \ref{Single-Ancilla LCU} with $\ket{\widetilde{\psi}}$ be the initial state and $\widetilde{O}$ be the observable. 
In the second step of Algorithm \ref{algorithm:random_LCU}, we can sample $V_1, V_2$ from distribution $\left\{U_{\boldsymbol{k}}, \frac{|c_{\boldsymbol{k}}|}{\|{\boldsymbol{c}}\|_1}\right\}$ with absorbing global phase. After the following steps in Algorithm \ref{algorithm:random_LCU}, we would get $\mu$ with
\begin{equation}
    \mathbb{E}(\mu)
    \approx\bra{\psi}f^\dagger(\boldsymbol A) Of(\boldsymbol A)\ket{\psi},
\end{equation}
and the following theorem holds.

\begin{theorem}
\label{theorem:Ancilla-efficient_convergence}
Let $\epsilon, \xi_1, \xi_2 \in(0,1)$ be some parameters. $U_{\boldsymbol{k}}$ is a $(1,a+2+\log n,0)$-block-encoding of $\prod_{j=1}^{n} p_j^\diamond\left(\frac{z_{j,k_j} I_m-A_j}{\alpha_j + |z_{j,k_j}|}\right)^\dagger$. Let $\ket{\psi}$ be the initial state and $\ket{\widetilde{\psi}} =\ket{0^{a+2+\log n}}\otimes\ket{\psi}$, $O$ be an observable and $\widetilde{O}=\ket{0^{a+3+\log n}}\bra{0^{a+3+\log n}}\otimes O$.
Suppose that $\left\|f(\boldsymbol A)\ket{\psi}-\sum_{\boldsymbol{k}} c_{\boldsymbol k}\prod_{j=1}^{n} p_j^\diamond\left(\frac{z_{j,k_j} I_m-A_j}{\alpha_j + |z_{j,k_j}|}\right)^\dagger\ket{\psi}\right\| \leq \xi_1$, where 
$$
\xi_1 \leq \frac{\epsilon}{6\|O\|\|f(\boldsymbol A)\|}.
$$
Furthermore, let
$$
T \geq \frac{\|O\|^2 \ln ( 2/ \xi_2)\|\boldsymbol c\|_1^4}{\epsilon^2}.
$$
Then, Algorithm \ref{algorithm:random_LCU} estimates $\mu$ such that
$$
\left|\mu-\bra{\psi}f(\boldsymbol A)^\dagger O f(\boldsymbol A)\ket{\psi}\right| \leq \epsilon,
$$
with probability at least $(1-\xi_2)^2$, using one ancilla qubit and $T$ repetitions of the quantum circuit in Algorithm \ref{algorithm:random_LCU}, Step 3.
\end{theorem}

We choose $\boldsymbol M$ and $\epsilon^\prime$ such that 
$$
\frac{1}{(2\pi)^n}\left(\prod_{\ell=1}^nl_\ell\gamma_\ell\right)\left(\sum_{j=1}^n\frac{l_j(B\gamma_j+BK_{\Gamma_j}+L)}{M_j}\right)\leq \frac{\xi_1}{2},
$$ 
and 
$$\left\|f_{\boldsymbol{M}}(\boldsymbol A)\ket{\psi}-\sum_{\boldsymbol{k}} c_{\boldsymbol k}\prod_{j=1}^{n} p_j^\diamond\left(\frac{z_{j,k_j} I_m-A_j}{\alpha_j + |z_{j,k_j}|}\right)^\dagger\ket{\psi}\right\| \leq \frac{\xi_1}{2}.$$
That means the degree of polynomial $p_j$ is 
$$
\mathcal{O} \left(\frac{1}{\delta_j}\mathrm{log}\frac{n\cdot\|\boldsymbol c\|_1}{\xi_1}\right)=\mathcal{O} \left(\gamma_j(\alpha_j+R_j)\mathrm{log}\frac{nB\prod _{\ell=1}^nl_\ell\gamma_\ell\|O\|\|f(\boldsymbol A)\|}{\epsilon}\right).
$$ 
The degree of polynomial $p$ is the same as the number of times QSVT visits $U_A$.
For each time we need $a$ ancilla qubits for $U_A$, $1$ for inner layer LCU, $1$ for QSVT, $\log n+1$ for compression gadget, and $1$ for outer layer random LCU. Therefore, we have the following theorem:

\begin{theorem}
Let $U_{A_j}$ be an $(\alpha_j,a,0)$-block-encoding of $A_j$ for $j=1,\cdots,n$. Under the assumption that for every $j=1,\ldots,n$ and every $z\in\Gamma_j$, \(\|(zI_m-A_j)^{-1}\|\le \gamma_j \), let
$$
T =\mathcal{O} \left(\frac{\|O\|^2 \ln (1 / \xi_2)\left(B\prod_{j=1}^n\frac{2\gamma_j l_j(1+R_j/\alpha_j)}{3\pi}\right)^4}{\epsilon^2 }\right),
$$ we can solve the problem defined by Definition \ref{definition:estimation_problem} with probability at least $(1-\xi_2)^2$ using a quantum circuit with
$$
\mathcal{O} \left(\gamma_j(\alpha_j+R_j)\mathrm{log}\frac{nB\prod _{\ell=1}^nl_\ell\gamma_\ell\|O\|\|f(\boldsymbol A)\|}{\epsilon}\right)
$$ queries to $U_{A_j}$, $1$ query to $U_{\psi}$ and $a+\log n+4$ ancilla qubits, and repeat this circuit $T$ times.

\label{theorem:Ancilla-efficient_complexity_contour}
\end{theorem}

\section{Fourier transform} \label{sec:fourier}

In this section, we introduce a quantum algorithm on multi matrix function based on Fourier transform for high-dimensional functions. Suppose $f(\boldsymbol{x})\in H^{s}(\mathbb{R}^n)\cap L^{1}(\mathbb{R}^n)$, where $\boldsymbol{x}=(x_1,x_2,\cdots,x_n)\in\mathbb{R}^n$ and $s>\frac{n}{2}$. Define the Fourier transform by
\begin{equation}
\label{equation:fourier_transform}
\widehat{f}(\boldsymbol{\lambda})=\frac{1}{(2 \pi)^{\frac{n}{2}}} \int_{\mathbb{R}^{n}} f(\boldsymbol{x}) \mathrm{e}^{-\mathrm{i} \boldsymbol{\lambda} \cdot \boldsymbol{x}} \mathrm{~d} \boldsymbol{x}.
\end{equation}
Then we have the Fourier inversion formula
\begin{equation}
\left(\widehat{f}(\boldsymbol{\lambda})\right)^{\vee}=\lim _{M \rightarrow+\infty} \frac{1}{(2 \pi)^{\frac{n}{2}}} \int_{|\boldsymbol{\lambda}|_\infty \leq M} \widehat{f}(\boldsymbol{\lambda}) \mathrm{e}^{\mathrm{i} \boldsymbol{\lambda} \cdot \boldsymbol{x}} \mathrm{~d} \boldsymbol{\lambda}=f(\boldsymbol{x}).
\end{equation}
It can be proved that this conclusion holds for the commutative Hermitian matrices. In fact, we have the following theorem.

\begin{theorem}
Suppose $\boldsymbol{H} =(H_1,H_2,\cdots,H_n)$ is an $n$-tuple of commuting Hermitian matrices. Then for any $f(\boldsymbol x)\in H^s(\mathbb{R}^n)\cap L^1(\mathbb{R}^n)$ with its Fourier transform $\widehat{f}(\boldsymbol\lambda)$ given by \cref{equation:fourier_transform}, we have that the function defined in \cref{definition:Multi-matrix_function} can be implemented by
\begin{equation}
\label{equation:matrix_fourier_transform}
    f(\boldsymbol{H})=\lim _{M \rightarrow+\infty} \frac{1}{(2 \pi)^{\frac{n}{2}}} \int_{|\boldsymbol{\lambda}|_\infty \leq M} \widehat{f}(\boldsymbol{\lambda}) \mathrm{e}^{\mathrm{i} \boldsymbol{\lambda} \cdot \boldsymbol{H}} \mathrm{~d} \boldsymbol{\lambda}.
\end{equation}

\end{theorem}
\begin{proof}
Since $H_1,H_2,\cdots,H_n$ are commuting Hermitian matrices, they can be diagonalized by the same unitary. Suppose that $H_j=U^\dagger D_j U$, where $U$ is a unitary matrix and $D_j=\operatorname{diag}(d_{j1},d_{j2},\cdots,d_{jm})$ then 
\begin{equation}
\begin{aligned}
\mathrm{e}^{\mathrm{i} \boldsymbol{\lambda} \cdot \boldsymbol{H}}
=&\mathrm{e}^{\mathrm{i} \sum_{j=1}^n\lambda_jH_j}=\prod_{j=1}^n\mathrm{e}^{\mathrm{i} \lambda_jH_j}=\prod_{j=1}^nU^\dagger\mathrm{e}^{\mathrm{i} \lambda_jD_j}U=U^\dagger\left(\prod_{j=1}^n \mathrm{e}^{\mathrm{i} \lambda_jD_j}\right)U\\
=&U^\dagger\prod_{j=1}^n\operatorname{diag}\left(\mathrm{e}^{\mathrm{i} \lambda_jd_{j1}},\mathrm{e}^{\mathrm{i} \lambda_jd_{j2}},\cdots,\mathrm{e}^{\mathrm{i} \lambda_jd_{jm}}\right)U\\
=&U^\dagger\operatorname{diag}\left(\prod_{j=1}^n\mathrm{e}^{\mathrm{i} \lambda_jd_{j1}},\prod_{j=1}^n\mathrm{e}^{\mathrm{i} \lambda_jd_{j2}},\cdots,\prod_{j=1}^n\mathrm{e}^{\mathrm{i} \lambda_jd_{jm}}\right)U\\
=&U^\dagger\operatorname{diag}\left(\mathrm{e}^{\mathrm{i} \boldsymbol{\lambda} \cdot \boldsymbol{d}_{1}}, \mathrm{e}^{\mathrm{i} \boldsymbol{\lambda} \cdot \boldsymbol{d}_{2}},\cdots,\mathrm{e}^{\mathrm{i} \boldsymbol{\lambda} \cdot \boldsymbol{d}_{m}}\right)U.
\end{aligned}
\end{equation}
Here $\boldsymbol{d}_k=(d_{1k},d_{2k},\cdots,d_{nk})\in\mathbb{R}^n$. Therefore, the RHS of \cref{equation:matrix_fourier_transform} can be written as
\begin{equation}
\begin{aligned}
&\lim _{M \rightarrow+\infty} \frac{1}{(2 \pi)^{\frac{n}{2}}} \int_{|\boldsymbol{\lambda}|_\infty \leq M} \widehat{f}(\boldsymbol{\lambda}) \mathrm{e}^{\mathrm{i} \boldsymbol{\lambda} \cdot \boldsymbol{H}} \mathrm{~d} \boldsymbol{\lambda}\\
=&\lim _{M \rightarrow+\infty} \frac{1}{(2 \pi)^{\frac{n}{2}}} \int_{|\boldsymbol{\lambda}|_\infty \leq M} \widehat{f}(\boldsymbol{\lambda}) U^\dagger\operatorname{diag}\left(\mathrm{e}^{\mathrm{i} \boldsymbol{\lambda} \cdot \boldsymbol{d}_{1}}, \mathrm{e}^{\mathrm{i} \boldsymbol{\lambda} \cdot \boldsymbol{d}_{2}},\cdots,\mathrm{e}^{\mathrm{i} \boldsymbol{\lambda} \cdot \boldsymbol{d}_{m}}\right)U \mathrm{~d} \boldsymbol{\lambda}\\
=&U^\dagger\operatorname{diag}\left(\lim _{M \rightarrow+\infty} \frac{1}{(2 \pi)^{\frac{n}{2}}} \int_{|\boldsymbol{\lambda}|_\infty \leq M} \widehat{f}(\boldsymbol{\lambda}) \mathrm{e}^{\mathrm{i} \boldsymbol{\lambda} \cdot \boldsymbol{d}_1} \mathrm{~d} \boldsymbol{\lambda},\cdots,\lim _{M \rightarrow+\infty} \frac{1}{(2 \pi)^{\frac{n}{2}}} \int_{|\boldsymbol{\lambda}|_\infty \leq M} \widehat{f}(\boldsymbol{\lambda}) \mathrm{e}^{\mathrm{i} \boldsymbol{\lambda} \cdot \boldsymbol{d}_m} \mathrm{~d} \boldsymbol{\lambda}\right)U\\
=&U^\dagger\operatorname{diag}\left(f(\boldsymbol{d}_1),\cdots,f(\boldsymbol{d}_m)\right)U\\
=&f\left(\boldsymbol{H}\right).
\end{aligned}
\end{equation}
\end{proof}

We can propose a quantum algorithm through this formula. However, when we use discrete point summation to approximate the integral, we cannot truly make $M$ infinite. Therefore, we give the error bound between the global integral and the local integral when the Fourier transform of $D^\alpha f$ is integrable on $\mathbb{R}^n$ for any $|\alpha|\leq k$, $k\in\mathbb{N}^+$. A sufficient condition for this assumption is $f(\boldsymbol{x})\in H^s(\mathbb{R}^n)\cap L^1(\mathbb{R}^n)$ and $s>k+\frac{n}{2}$. When $k=1$, we can control the error by choosing a sufficiently large $M$. As $k$ becomes larger, the requirements for the value of $M$ become less stringent.

\begin{lemma}
\label{lemma:fourier_truncation_error}
Suppose $f(\boldsymbol{x})\in H^s(\mathbb{R}^n)\cap L^1(\mathbb{R}^n)$ and $s>k+\frac{n}{2}$. For \cref{equation:matrix_fourier_transform}, we have the following error estimate:
\begin{equation}
    \left\|f(\boldsymbol{H})- \frac{1}{(2 \pi)^{\frac{n}{2}}} \int_{|\boldsymbol{\lambda}|_\infty \leq M} \widehat{f}(\boldsymbol{\lambda}) \mathrm{e}^{\mathrm{i} \boldsymbol{\lambda} \cdot \boldsymbol{H}} \mathrm{~d} \boldsymbol{\lambda}\right\|\leq\frac{1}{(2 \pi)^{\frac{n}{2}}}\cdot\frac{1}{M^k}\sum_{j=1}^n\left\|\left(\frac{\partial^kf}{\partial x_j^k}\right)^{\wedge}\right\|_{L^1(\mathbb{R}^n)}.
\end{equation}
\end{lemma}
\begin{proof}
By the Cauchy--Schwarz inequality,
\begin{equation}
\begin{aligned}
\left\|
\left(\frac{\partial^k f}{\partial x_j^k}\right)^\wedge
\right\|_{L^1(\mathbb{R}^n)}
&=
\int_{\mathbb{R}^n}
|\lambda_j|^k
|\widehat{f}(\boldsymbol{\lambda})|
\,\mathrm{d}\boldsymbol{\lambda}\\
&\leq
\left(
\int_{\mathbb{R}^n}
(1+\|\boldsymbol{\lambda}\|^2)^{-(s-k)}
\,\mathrm{d}\boldsymbol{\lambda}
\right)^{1/2}
\|f\|_{H^s(\mathbb{R}^n)}
<\infty,
\qquad s>k+\frac{n}{2}.
\end{aligned}
\end{equation}
According to the property of high-dimensional Fourier transform,
\begin{equation}
    \left(\frac{\partial^kf}{\partial x_j^k}\right)^{\wedge}=(\mathrm{i}\lambda_j)^k\widehat{f}.
\end{equation}
Therefore,
\begin{equation}
    \int_{|\lambda_j|\geq M}\left|\widehat{f}(\boldsymbol{\lambda})\right|\mathrm{~d} \boldsymbol{\lambda}=\int_{|\lambda_j|\geq M}\left|\frac{1}{\lambda_j^k}\left(\frac{\partial^kf}{\partial x_j^k}\right)^{\wedge}(\boldsymbol{\lambda})\right|\mathrm{~d} \boldsymbol{\lambda}\leq \frac{1}{M^k}\int_{|\lambda_j|\geq M}\left|\left(\frac{\partial^kf}{\partial x_j^k}\right)^{\wedge}(\boldsymbol{\lambda})\right|\mathrm{~d} \boldsymbol{\lambda},
\end{equation}
and we can have the bound as 
\begin{align}
&\left\|f(\boldsymbol{H})- \frac{1}{(2 \pi)^{\frac{n}{2}}} \int_{|\boldsymbol{\lambda}|_\infty \leq M} \widehat{f}(\boldsymbol{\lambda}) \mathrm{e}^{\mathrm{i} \boldsymbol{\lambda} \cdot \boldsymbol{H}} \mathrm{~d} \boldsymbol{\lambda}\right\|\notag\\
\leq&\frac{1}{(2 \pi)^{\frac{n}{2}}}\left\|\left(\int_\mathbb{R}\int_\mathbb{R}\cdots\int_\mathbb{R}-\int_{-M}^M\int_\mathbb{R}\cdots\int_\mathbb{R}\right)\widehat{f}(\boldsymbol{\lambda}) \mathrm{e}^{\mathrm{i} \boldsymbol{\lambda} \cdot \boldsymbol{H}} \mathrm{~d} \boldsymbol{\lambda}\right\|+\notag\\
&\frac{1}{(2 \pi)^{\frac{n}{2}}}\left\|\left(\int_{-M}^M\int_\mathbb{R}\cdots\int_\mathbb{R}-\int_{-M}^M\int_{-M}^M\cdots\int_\mathbb{R}\right)\widehat{f}(\boldsymbol{\lambda}) \mathrm{e}^{\mathrm{i} \boldsymbol{\lambda} \cdot \boldsymbol{H}} \mathrm{~d} \boldsymbol{\lambda}\right\|+\cdots+\notag\\
&\frac{1}{(2 \pi)^{\frac{n}{2}}}\left\|\left(
\int_{-M}^M\cdots\int_{-M}^M\int_\mathbb{R}-\int_{-M}^M\cdots\int_{-M}^M\int_{-M}^M\right)\widehat{f}(\boldsymbol{\lambda}) \mathrm{e}^{\mathrm{i} \boldsymbol{\lambda} \cdot \boldsymbol{H}} \mathrm{~d} \boldsymbol{\lambda}\right\|\notag\\
\leq&\frac{1}{(2 \pi)^{\frac{n}{2}}}\left(\sum_{j=1}^n\int_{[-M,M]^{j-1}}\int_{|\lambda_j|\geq M}\int_{\mathbb{R}^{n-j}}|\widehat{f}(\boldsymbol{\lambda}) | \mathrm{~d} \boldsymbol{\lambda}\right)\notag\\
\leq&\frac{1}{(2 \pi)^{\frac{n}{2}}}\left(\sum_{j=1}^n\frac{1}{M^k}\int_{[-M,M]^{j-1}}\int_{|\lambda_j|\geq M}\int_{\mathbb{R}^{n-j}}\left|\left(\frac{\partial^kf}{\partial x_j^k}\right)^{\wedge}(\boldsymbol{\lambda})\right| \mathrm{~d} \boldsymbol{\lambda}\right)\notag\\
\leq&\frac{1}{(2 \pi)^{\frac{n}{2}}}\cdot\frac{1}{M^k}\sum_{j=1}^n\left\|\left(\frac{\partial^kf}{\partial x_j^k}\right)^{\wedge}\right\|_{L^1(\mathbb{R}^n)}.
\end{align}
\end{proof}

Next, we take $N$ discrete points equally spaced in each interval $[-M,M]$, and approximate $f(\boldsymbol{H})$ by its Riemann sum $f_N(\boldsymbol{H})$ defined as 
\begin{equation}
f_N(\boldsymbol{H}) = \frac{(2M)^n}{(2 \pi)^\frac{n}{2}N^n}\sum_{k_1=0}^{N-1}\cdots\sum_{k_n=0}^{N-1} \widehat{f}(\lambda_{1,k_1},\cdots,\lambda_{n,k_n})\prod_{j=1}^{n}\mathrm{e}^{\mathrm{i} \lambda_{j,k_j} H_j}.
\label{function:discrete_fourier_transform}
\end{equation}
Here $\lambda_{j,k_j}=-M+\frac{2k_jM}{N}$. Discrete error can be given by the following lemma:
\begin{lemma}
\label{lemma:fourier_discretization_error}
Suppose $f(\boldsymbol{x})\in H^s(\mathbb{R}^n)\cap L^1(\mathbb{R}^n)$, $s>k+\frac{n}{2}$ and for $j=1,\cdots,n$, $x_jf(\boldsymbol{x})\in L_{1}(\mathbb{R}^n)$ then we have the bound
\begin{equation}
    \left\|\frac{1}{(2 \pi)^{\frac{n}{2}}} \int_{|\boldsymbol{\lambda}|_\infty \leq M} \widehat{f}(\boldsymbol{\lambda}) \mathrm{e}^{\mathrm{i} \boldsymbol{\lambda} \cdot \boldsymbol{H}} \mathrm{~d} \boldsymbol{\lambda}-f_N(\boldsymbol{H})\right\|\leq\frac{(2M)^{n+1}}{(2 \pi)^{n}N}\left(\left\|f\right\|_{L_1} \sum_{j=1}^n\|H_j\|+ \sum_{j=1}^n\left\|x_jf\right\|_{L_1}\right).
\end{equation}    
\end{lemma}
\begin{proof}
    
We can provide an estimate of the discrete error

\begin{align}
\label{equation:fourier_discretization_error}
&\left\|\frac{1}{(2 \pi)^{\frac{n}{2}}} \int_{|\boldsymbol{\lambda}|_\infty \leq M} \widehat{f}(\boldsymbol{\lambda}) \mathrm{e}^{\mathrm{i} \boldsymbol{\lambda} \cdot \boldsymbol{H}} \mathrm{~d} \boldsymbol{\lambda}-\frac{(2M)^n}{(2 \pi)^\frac{n}{2}N^n}\sum_{k_1=0}^{N-1}\cdots\sum_{k_n=0}^{N-1} \widehat{f}(\lambda_{1,k_1},\cdots,\lambda_{n,k_n})\prod_{j=1}^{n}\mathrm{e}^{\mathrm{i} \lambda_{j,k_j} H_j}\right\|\notag\\
\leq&\sum_{k_1=0}^{N-1}\cdots\sum_{k_n=0}^{N-1} \frac{1}{(2 \pi)^{\frac{n}{2}}}\left\|\int_{\lambda_{1,k_1}}^{\lambda_{1,k_1+1}}\cdots\int_{\lambda_{n,k_n}}^{\lambda_{n,k_n+1}}\left(\widehat{f}(\boldsymbol{\lambda}) \mathrm{e}^{\mathrm{i} \boldsymbol{\lambda} \cdot \boldsymbol{H}}-\widehat{f}(\lambda_{1,k_1},\cdots,\lambda_{n,k_n})\prod_{j=1}^{n}\mathrm{e}^{\mathrm{i} \lambda_{j,k_j} H_j}\right)\mathrm{~d} \boldsymbol{\lambda}\right\|\notag\\
\leq&\sum_{k_1=0}^{N-1}\cdots\sum_{k_n=0}^{N-1} \frac{1}{(2 \pi)^{\frac{n}{2}}}\int_{\lambda_{1,k_1}}^{\lambda_{1,k_1+1}}\cdots\int_{\lambda_{n,k_n}}^{\lambda_{n,k_n+1}}\left\|\widehat{f}(\boldsymbol{\lambda}) \left(\mathrm{e}^{\mathrm{i} \boldsymbol{\lambda} \cdot \boldsymbol{H}}-\prod_{j=1}^{n}\mathrm{e}^{\mathrm{i} \lambda_{j,k_j} H_j}\right)\right\|\mathrm{~d} \boldsymbol{\lambda}+\notag\\
&\sum_{k_1=0}^{N-1}\cdots\sum_{k_n=0}^{N-1} \frac{1}{(2 \pi)^{\frac{n}{2}}}\int_{\lambda_{1,k_1}}^{\lambda_{1,k_1+1}}\cdots\int_{\lambda_{n,k_n}}^{\lambda_{n,k_n+1}}\left\|\left(\widehat{f}(\boldsymbol{\lambda})-\widehat{f}(\lambda_{1,k_1},\cdots,\lambda_{n,k_n})\right)\prod_{j=1}^{n}\mathrm{e}^{\mathrm{i} \lambda_{j,k_j} H_j}\right\|\mathrm{~d} \boldsymbol{\lambda}\notag\\
\leq&\sum_{k_1=0}^{N-1}\cdots\sum_{k_n=0}^{N-1} \frac{1}{(2 \pi)^{\frac{n}{2}}} \int_{\lambda_{1,k_1}}^{\lambda_{1,k_1+1}}\cdots\int_{\lambda_{n,k_n}}^{\lambda_{n,k_n+1}}\left(\sup_{\boldsymbol{\lambda}}\left|\widehat{f}(\boldsymbol{\lambda})\right| \sum_{j=1}^n\frac{2M}{N}\|H_j\|+ \sum_{j=1}^n\frac{2M}{N}\sup_{\boldsymbol{\lambda}}\left|\frac{\partial}{\partial \lambda_j}\widehat{f}(\boldsymbol{\lambda})\right|\right)\mathrm{~d} \boldsymbol{\lambda}\notag\\
=&\frac{(2M)^{n+1}}{(2 \pi)^{\frac{n}{2}}N}\left(\sup_{\boldsymbol{\lambda}}\left|\widehat{f}(\boldsymbol{\lambda})\right| \sum_{j=1}^n\|H_j\|+ \sum_{j=1}^n\sup_{\boldsymbol{\lambda}}\left|\frac{\partial}{\partial \lambda_j}\widehat{f}(\boldsymbol{\lambda})\right|\right)\notag\\
=&\frac{(2M)^{n+1}}{(2 \pi)^{n}N}\left(\sup_{\boldsymbol{\lambda}}\left|\int_{\mathbb{R}^n} f(\boldsymbol{x})e^{-i\boldsymbol{\lambda}\cdot \boldsymbol{x}}\,d\boldsymbol{x}\right| \sum_{j=1}^n\|H_j\|+ \sum_{j=1}^n\sup_{\boldsymbol{\lambda}}\left|\int_{\mathbb{R}^n}(-i x_j) f(\boldsymbol{x})e^{-i\boldsymbol{\lambda}\cdot \boldsymbol{x}}\,d\boldsymbol{x}\right|\right)\notag\\
\leq&\frac{(2M)^{n+1}}{(2 \pi)^{n}N}\left(\left\|f\right\|_{L_1} \sum_{j=1}^n\|H_j\|+ \sum_{j=1}^n\left\|x_jf\right\|_{L_1}\right).
\end{align}
\end{proof}

Combining \cref{lemma:fourier_discretization_error} and \cref{lemma:fourier_truncation_error}, We can give the error bound between the function $f(\boldsymbol{H})$ and its discrete estimation $f_N(\boldsymbol{H})$.
\begin{equation}
\label{equation:Discrete_integration_error_by_fourier_transform}
    \left\|f(\boldsymbol{H})-f_N(\boldsymbol{H})\right\|\leq\frac{(2M)^{n+1}}{(2 \pi)^{n}N}\left(\left\|f\right\|_{L_1} \sum_{j=1}^n\|H_j\|+ \sum_{j=1}^n\left\|x_jf\right\|_{L_1}\right)+\frac{1}{(2 \pi)^{\frac{n}{2}}}\cdot\frac{1}{M^k}\sum_{j=1}^n\left\|\left(\frac{\partial^kf}{\partial x_j^k}\right)^{\wedge}\right\|_{L^1(\mathbb{R}^n)}.
\end{equation}

\subsection{Quantum algorithm}
\label{subsection:Fourier Transform Algorithm}

In this part, we introduce a quantum algorithm to approximate $f_N(\boldsymbol{H})$. Note that for $j=1,\cdots,n$ we have block encoding \( U_{H_j} \) of the matrix \( H_j \). 
Then
\begin{equation}
    U_{H_j}\ket{0^{a}}\ket{\psi} = \frac{1}{\alpha_j}\ket{0^{a}}H_j\ket{\psi} +\ket{\perp}.
\end{equation}
Recall that $\lambda_{j,k_j}=-M+\frac{2k_jM}{N}$, $|\lambda_{j,k_j}|\leq M$ for any $j$ and $k_j$. Let $V_j$ be controlled rotations satisfying
\begin{equation}
    V_j\ket{k_j}\ket{0}_{V_j} = \ket{k_j}\left(\frac{\lambda_{j,k_j}}{M}\ket{0}+\sqrt{1-\left(\frac{\lambda_{j,k_j}}{M}\right)^2}\ket{1}\right),
    \label{operator:control_rotation_fourier}
\end{equation}
then
\begin{equation}
    \left(V_j \otimes I_{a}\otimes I_{m}\right) (I_{k_j}\otimes I_{V_j}\otimes U_{H_j})\ket{k_j} \ket{0}_{V_j}\ket{0^{a}}\ket{\psi} = \ket{k_j}\ket{0}_{V_j}\ket{0^{a}}\frac{\lambda_{j,k_j}}{\alpha_j M}H_j\ket{\psi}+\ket{\widetilde{\perp}}, 
     \label{operator:control_block_encoding_fourier}
\end{equation}
This is an $(\alpha_j, a+1)$-block-encoding of Hamiltonian $\frac{\lambda_{j,k_j}}{M}H_j$ controlled by $\ket{k_j}$, where $\ket{\widetilde{\perp}}$ means vectors where the second or third register is not zero. For any Hamiltonian $H$, let $U_{H}$ be an $(\alpha_j, a+1)$-block-encoding of $H$ and $t\in\mathbb{R}$, then \cite[Corollary 60]{GilyenSuLowEtAl2018} shows that we can construct an $(1, a+3, \epsilon^\prime)$-block-encoding of $e^{\mathrm{i} t H}$ with using the unitary $U_{H}$ a total number of times
$$
\mathcal{O}\left(\alpha_j|t|+\frac{\log (1 / \epsilon^\prime)}{\log (e+\log (1 / \epsilon^\prime) /(\alpha_j|t|))}\right).
$$
The construction of this circuit has nothing to do with the information of $H$, therefore similar to \cref{fig:circuit_Control_QSVT}, we can apply the control circuit \cref{operator:control_block_encoding_fourier} to $U_H$ and let $t=M$ to get an $(1, a+3, \epsilon^\prime)$-block-encoding of $e^{\mathrm{i}  \lambda_{j,k_j}H_j}$ controlled by $\ket{k_j}$.

Considering this controlled block encoding as $\widetilde{U}_{k_j}$, we can also get a block encoding of $\prod_{j=1}^{n}\mathrm{e}^{\mathrm{i} \lambda_{j,k_j} H_j}$ with compression gadget according to \cref{fig:circuit_multi_ctrl}, and it can also be the select oracle SEL in LCU algorithm. The prepare oracle $C$ is a unitary such that
\begin{equation}
    C\ket{0}\cdots\ket{0}=\frac{1}{\sqrt{\|\boldsymbol c\|_1}} \sum_{k_1=0}^{N-1}\cdots\sum_{k_n=0}^{N-1} \sqrt{c_{\boldsymbol k}}\ket{k_1}\cdots\ket{k_n},
\end{equation}
where $\boldsymbol c=(c_{(0,\cdots,0)},c_{(0,\cdots,1)},\cdots,c_{(N-1,\cdots,N-1)})$, and
\begin{equation}
c_{\boldsymbol k}=c_{(k_1,\cdots,k_n)}=\frac{(2M)^n}{(2 \pi)^\frac{n}{2}N^n} \widehat{f}(\lambda_{1,k_1},\cdots,\lambda_{n,k_n}).
\end{equation}
Then by \cref{lemma:lcu}, the operator 
\begin{equation}\label{operator:total_circuit_fourier}
     \left(\widetilde{C} \otimes I\right)\mathrm{SEL}\left(C\otimes I\right)
\end{equation}
gives a block encoding of $$\sum_{k_1=0}^{N-1}\cdots\sum_{k_n=0}^{N-1} \frac{(2M)^n}{(2 \pi)^\frac{n}{2}N^n} \widehat{f}(\lambda_{1,k_1},\cdots,\lambda_{n,k_n})\prod_{j=1}^{n}\mathrm{e}^{\mathrm{i} \lambda_{j,k_j} H_j}, $$ which is an approximation of $f_{N}(\boldsymbol H)$.

\subsection{Complexity analysis}

This part we analyze the complexity the algorithm given in \cref{subsection:Fourier Transform Algorithm}.

In the previous discussion, we showed that in order to construct $\widetilde{U}_{k_j}$, which is an $(1, a+3, \epsilon^\prime)$-block-encoding of $\mathrm{e}^{\mathrm{i} \lambda_{j,k_j} H_j}$ controlled by $k_j$, we need
\begin{equation}
\mathcal{O}\left(\alpha_jM+\frac{\log (1 / \epsilon^\prime)}{\log (e+\log (1 / \epsilon^\prime) /(\alpha_jM))}\right).
\end{equation}
calls to $U_{H_j}$. Suppose
\begin{equation}
\widetilde{U}_{k_j}\ket{k_j}\ket{0^{a+3}}\ket{\psi}=\ket{k_j}\ket{0^{a+3}}g_{\lambda_{j,k_j}}(H_j)\ket{\psi}+\ket{\perp},
\end{equation}
then $\|g_{\lambda_{j,k_j}}(H_j)\|\leq1$ and $\|g_{\lambda_{j,k_j}}(H_j)-e^{\mathrm{i} \lambda_{j,k_j} H_j}\|\leq\epsilon^\prime.$ Let
\begin{equation}
    \widetilde{f}_N(\boldsymbol H)=\sum_{k_1=0}^{N-1}\cdots\sum_{k_n=0}^{N-1} c_{\boldsymbol k}\prod_{j=1}^n g_{\lambda_{j,k_j}}(H_j),
\end{equation}
we know that
\begin{equation}
\label{equation:fourier_QSVT_error}
\left\|\widetilde{f}_N(\boldsymbol H) -f_{N}(\boldsymbol H)\right\|
    \leq\left(\sum_{k_1=0}^{N-1}\cdots\sum_{k_n=0}^{N-1} \|c_{\boldsymbol k}\|\right)\left\| \prod_{j=1}^ng_{\lambda_{j,k_j}}(H_j)-\prod_{j=1}^ne^{\mathrm{i} \lambda_{j,k_j} H_j}\right\|
    \leq \|\boldsymbol c\|_1 n\epsilon^\prime.
\end{equation}

Also choose proper $M$, $N$ so that $$\frac{(2M)^{n+1}}{(2 \pi)^{n}N}\left(\left\|f\right\|_{L_1} \sum_{j=1}^n\|H_j\|+ \sum_{j=1}^n\left\|x_jf\right\|_{L_1}\right)+\frac{1}{(2 \pi)^{\frac{n}{2}}}\cdot\frac{1}{M^k}\sum_{j=1}^n\left\|\left(\frac{\partial^kf}{\partial x_j^k}\right)^{\wedge}\right\|_{L^1(\mathbb{R}^n)}\leq \|\boldsymbol c\|_1 n\epsilon^\prime,$$
then together with \cref{equation:Discrete_integration_error_by_fourier_transform},
\begin{equation}
\left\|\widetilde{f}_N(\boldsymbol H) -f(\boldsymbol H)\right\| 
    \leq 2\|\boldsymbol c\|_1 n \epsilon^\prime.
\end{equation}
Due to that $\|\frac{u}{\|u\|}-\frac{v}{\|v\|}\|\leq2\frac{\|u-v\|}{\|u\|}$, we need
\begin{equation}
    \left\|\frac{f(\boldsymbol H)\ket{\psi}}{\|f(\boldsymbol H)\ket{\psi}\|}-\frac{\widetilde{f}_N(\boldsymbol H)\ket{\psi}}{\|\widetilde{f}_N(\boldsymbol H)\ket{\psi}\|}\right\|
    \leq \frac{4\|\boldsymbol c\|_1 n \epsilon^\prime}{\|f(\boldsymbol H)\ket{\psi}\|}\leq\epsilon.
\end{equation}

Therefore let
\begin{equation}
    \epsilon^\prime=\frac{\|f(\boldsymbol H)\ket{\psi}\|\epsilon}{4\|\boldsymbol c\|_1 n}
\end{equation}
and
\begin{equation}
\frac{(2M)^{n+1}}{(2 \pi)^{n}N}\left(\left\|f\right\|_{L_1} \sum_{j=1}^n\|H_j\|+ \sum_{j=1}^n\left\|x_jf\right\|_{L_1}\right)+\frac{1}{(2 \pi)^{\frac{n}{2}}}\cdot\frac{1}{M^k}\sum_{j=1}^n\left\|\left(\frac{\partial^kf}{\partial x_j^k}\right)^{\wedge}\right\|_{L^1(\mathbb{R}^n)}\leq\frac{\|f(\boldsymbol H)\ket{\psi}\|\epsilon}{4},
\end{equation}
then \cref{operator:total_circuit_fourier} is a block encoding of $f(\boldsymbol H)$ with block encoding factor $\|\boldsymbol{c}\|_1$ and error $\epsilon$. A possible choice of $M, N$ is
\begin{equation}
\label{parameter:M_fourier}
    M\geq\left(\frac{1}{(2 \pi)^{\frac{n}{2}}}\cdot\frac{8}{\|f(\boldsymbol H)\ket{\psi}\|\epsilon}\cdot\sum_{j=1}^n\left\|\left(\frac{\partial^kf}{\partial x_j^k}\right)^{\wedge}\right\|_{L^1(\mathbb{R}^n)}\right)^\frac{1}{k},
\end{equation}
\begin{equation}
\label{parameter:N_fourier}
    N\geq\frac{(2M)^{n+1}}{(2 \pi)^{n}}\cdot\left(\left\|f\right\|_{L_1} \sum_{j=1}^n\|H_j\|+ \sum_{j=1}^n\left\|x_jf\right\|_{L_1}\right)\cdot\frac{8}{\|f(\boldsymbol H)\ket{\psi}\|\epsilon}.
\end{equation}

According to \cref{fig:circuit_multi_ctrl}, we need $n\log N+\log n+a+4$ ancilla qubits.
The successful probability of this circuit is $\|\widetilde{f}_N(\boldsymbol H)\ket{\psi}\|^2/\|\boldsymbol c\|_1^2$, therefore with amplitude amplification, 
$\mathcal{O}(\|\boldsymbol c\|_1/\|\widetilde{f}_N(\boldsymbol H)\ket{\psi}\|)$
calls to \cref{operator:total_circuit_fourier} can make the algorithm succeed with probability at least $1/2$.

On the other hand,
\begin{equation}
    \|\boldsymbol c\|_1=\sum_{k_1=0}^{N-1}\cdots\sum_{k_n=0}^{N-1} \frac{(2M)^n}{(2 \pi)^\frac{n}{2}N^n} \left|\widehat{f}(\lambda_{1,k_1},\cdots,\lambda_{n,k_n})\right|
    = \mathcal{O} \left(\frac{1}{(2 \pi)^{\frac{n}{2}}} \int_{\mathbb{R}^n} \left|\widehat{f}(\boldsymbol{\lambda})\right| \mathrm{~d} \boldsymbol{\lambda}\right),
\end{equation}
and we use $\widetilde{f}_N(\boldsymbol H)$ to approximate $f(\boldsymbol H)$. For numerical format of sufficiently high precision, 
\begin{equation}
\begin{aligned}
\|\boldsymbol c\|_1\lesssim& \frac{1}{(2 \pi)^{\frac{n}{2}}} \int_{\mathbb{R}^n} \left|\widehat{f}(\boldsymbol{\lambda})\right| \mathrm{~d} \boldsymbol{\lambda}=
\frac{1}{(2 \pi)^{\frac{n}{2}}}\|\widehat{f}\|_{L^1(\mathbb{R}^n)}
\le
\frac{1}{(2 \pi)^{\frac{n}{2}}}\left(
\int_{\mathbb{R}^n}
(1+|\xi|^2)^{-s}\,d\xi
\right)^{1/2}
\|f\|_{H^s(\mathbb{R}^n)}\\
=&
\frac{1}{(2 \pi)^{\frac{n}{2}}}\left(
\pi^{n/2}
\frac{\Gamma\!\left(s-\frac{n}{2}\right)}
{\Gamma(s)}
\right)^{1/2}
\|f\|_{H^s(\mathbb{R}^n)},
\qquad
s>\frac{n}{2}.
\end{aligned}
\end{equation}
Therefore, $$\mathcal{O}(\|\boldsymbol c\|_1/\|\widetilde{f}_N(\boldsymbol H)\ket{\psi}\|)=\mathcal{O}\left(\frac{\left(
\pi^{n/2}
\frac{\Gamma\!\left(s-\frac{n}{2}\right)}
{\Gamma(s)}
\right)^{1/2}
\|f\|_{H^s(\mathbb{R}^n)}}{(2 \pi)^{\frac{n}{2}}\|f(\boldsymbol H)\ket{\psi}\|}\right),$$
and
\begin{equation}
\label{parameter:epsilon_prime_fourier}
    \frac{1}{\epsilon^\prime}=\mathcal{O}\left(\frac{2n\left(
\pi^{n/2}
\frac{\Gamma\!\left(s-\frac{n}{2}\right)}
{\Gamma(s)}
\right)^{1/2}
\|f\|_{H^s(\mathbb{R}^n)}}{ (2 \pi)^{\frac{n}{2}}\|f(\boldsymbol H)\ket{\psi}\|\epsilon}\right)
\end{equation}

Summarizing these results, we can get the following theorem.

\begin{theorem}[Complexity of the Fourier transform approach]
Let $U_{H_j}$ be an $(\alpha_j,a,0)$-block-encoding of $H_j$ for $j=1,\cdots,n$. Under the assumption that $f(\boldsymbol{x})\in H^s(\mathbb{R}^n)\cap L^1(\mathbb{R}^n)$, $s>k+\frac{n}{2}$ and $x_jf(\boldsymbol{x})\in L_{1}(\mathbb{R}^n)$ for $j=1,\cdots,n$, we can solve the problem defined by Definition \ref{definition:main_problem} using a quantum circuit with
$$\mathcal{O}\left(\left(\alpha_jM+\frac{\log (1 / \epsilon^\prime)}{\log (e+\log (1 / \epsilon^\prime) /(\alpha_jM))}\right)\cdot\frac{\left(
\pi^{n/2}
\frac{\Gamma\!\left(s-\frac{n}{2}\right)}
{\Gamma(s)}
\right)^{1/2}
\|f\|_{H^s(\mathbb{R}^n)}}{(2 \pi)^{\frac{n}{2}}\|f(\boldsymbol H)\ket{\psi}\|}\right)$$ queries to $U_{H_j}$, 
$$\mathcal{O}\left(\frac{\left(
\pi^{n/2}
\frac{\Gamma\!\left(s-\frac{n}{2}\right)}
{\Gamma(s)}
\right)^{1/2}
\|f\|_{H^s(\mathbb{R}^n)}}{(2 \pi)^{\frac{n}{2}}\|f(\boldsymbol H)\ket{\psi}\|}\right)$$ queries to $U_{\psi}$ and $n\log N+\log n+a+4$ ancilla qubits, where $\epsilon^\prime$ is given in \cref{parameter:epsilon_prime_fourier}, $M$ in \cref{parameter:M_fourier} and $N$ in \cref{parameter:N_fourier}. 

\label{theorem:fourier_transform_complexity}
\end{theorem}

It is worth noting that although in this theorem the query complexity contains terms relating to the exponential dependence of $n$, when the properties of $f$ are sufficiently good, the actual query complexity doesn't have the exponential dependence of $n$, see \cref{Appendix:Non-exponential dependence of Fourier transform}.

\subsection{Ancilla-efficient algorithm}
Similar to \cref{subsection:Ancilla-efficient_Algorithm_contour}, for any $j$ and $k_j$ we would implement
$e^{\mathrm{i}  \lambda_{j,k_j}H_j}$ with QSVT separately, and with compression gadget, we can get a block encoding of $\prod_{j=1}^ne^{\mathrm{i} \lambda_{j,k_j} H_j}$ with $a+3+\log n$ ancilla qubits, and we name this block encoding $U_{\boldsymbol{k}}$.

Same as above, we can let $\ket{\widetilde{\psi}}=\ket{0^{a+3+\log n}}\otimes\ket{\psi}$ and $\widetilde{O}=\ket{0^{a+3+\log n}}\bra{0^{a+3+\log n}}\otimes O$, then 
\begin{equation}
\bra{\widetilde{\psi}}U_{\boldsymbol{k}}^\dagger\widetilde{O}U_{\boldsymbol{k}}\ket{\widetilde{\psi}}
=\bra{\psi} \left(\prod_{j=1}^ne^{\mathrm{i} \lambda_{j,k_j} H_j}\right)^\dagger O \prod_{j=1}^ne^{\mathrm{i} \lambda_{j,k_j} H_j}\ket{\psi}.
\end{equation}
Therefore, we can use the algorithm in Section \ref{Single-Ancilla LCU} with $\ket{\widetilde{\psi}}$ be the initial state and $\widetilde{O}$ be the observable. 
In the second step of Algorithm \ref{algorithm:random_LCU}, we can sample $V_1, V_2$ from distribution $\left\{U_{\boldsymbol{k}}, \frac{|c_{\boldsymbol{k}}|}{\|{\boldsymbol{c}}\|_1}\right\}$ with absorbing global phase. After the following steps in Algorithm \ref{algorithm:random_LCU}, we would get $\mu$ with
\begin{equation}
    \mathbb{E}(\mu)
    \approx\bra{\psi}f^\dagger(\boldsymbol{H}) Of(\boldsymbol{H})\ket{\psi},
\end{equation}
and \cref{theorem:Ancilla-efficient_convergence} holds.

We choose $M,N$ and $\epsilon^\prime$ such that 
\begin{align*}
    & \quad \left\|f(\boldsymbol{H})-f_N(\boldsymbol{H})\right\| \\
    & \leq\frac{(2M)^{n+1}}{(2 \pi)^{n}N}\left(\left\|f\right\|_{L_1} \sum_{j=1}^n\|H_j\|+ \sum_{j=1}^n\left\|x_jf\right\|_{L_1}\right)+\frac{1}{(2 \pi)^{\frac{n}{2}}}\cdot\frac{1}{M^k}\sum_{j=1}^n\left\|\left(\frac{\partial^kf}{\partial x_j^k}\right)^{\wedge}\right\|_{L^1(\mathbb{R}^n)}\\
    & \leq \frac{\xi_1}{2},
\end{align*}
and 
$$\left\|\widetilde{f}_N(\boldsymbol H) -f_{N}(\boldsymbol H)\right\| \leq \|\boldsymbol c\|_1 n\epsilon^\prime \leq \frac{\xi_1}{2}.$$
Then we need
\begin{equation}
\mathcal{O}\left(\alpha_jM+\frac{\log (1 / \epsilon^\prime)}{\log (e+\log (1 / \epsilon^\prime) /(\alpha_jM))}\right).
\end{equation}
calls to $U_{H_j}$.
For each time we need $a$ ancilla qubits for $U_{H_j}$, $2$ for QSVT, $\log n+1$ for compression gadget, and $1$ for LCU. Therefore, we have the following theorem. 

\begin{theorem}
Let $U_{H_j}$ be an $(\alpha_j,a,0)$-block-encoding of $H_j$ for $j=1,\cdots,n$. Let
$$
T =\mathcal{O} \left(\frac{\|O\|^2 \ln (1 / \xi_2)\left(\frac{1}{(2 \pi)^{\frac{n}{2}}} \left(
\pi^{n/2}
\frac{\Gamma\!\left(s-\frac{n}{2}\right)}
{\Gamma(s)}
\right)^{1/2}
\|f\|_{H^s(\mathbb{R}^n)}\right)^4}{\epsilon^2 }\right),
$$ we can solve the problem defined by Definition \ref{definition:estimation_problem} with probability at least $(1-\xi_2)^2$ using a quantum circuit with
$$
\mathcal{O}\left(\alpha_jM+\frac{\log (1 / \epsilon^\prime)}{\log (e+\log (1 / \epsilon^\prime) /(\alpha_jM))}\right)
$$ queries to $U_{H_j}$, $1$ query to $U_{\psi}$ and $a+\log n+4$ ancilla qubits, and repeat this circuit $T$ times.

\label{theorem:Ancilla-efficient_complexity_fourier}
\end{theorem}

\section{Application of matrix polynomial}\label{sec:application}

In this section, we discuss quantum algorithms for multivariable matrix polynomials. 
A straightforward implementation of a multivariable matrix polynomial constructs block-encodings of its monomials by multiplication and combines them using LCU. For example, given an \((\alpha_1,a,0)\)-block-encoding of \(A_1\), multiplying \(d\) copies with separate ancilla registers yields an \((\alpha_1^d,ad,0)\)-block-encoding of \(A_1^d\). The compression gadget \cref{lemma:compression_gadget} reduces the ancilla requirement to \(a+O(\log d)\), while preserving the normalization factor \(\alpha_1^d\). Applying this construction to \(|\psi\rangle\) gives the success probability $p_{\mathrm{succ}}=\|A_1^d|\psi\rangle\|^2/\alpha_1^{2d}.$
Since each application of the product requires \(d\) queries to \(U_{A_1}\), amplitude amplification gives a total query cost of $\mathcal{O}\!\left({d\,\alpha_1^d}/{\|A_1^d|\psi\rangle\|}\right).$
Thus, even when \(\|A_1\|\leq1\), a fixed normalization factor \(\alpha_1>1\) introduces an exponential overhead in \(d\) for this construction.
The complexities of our methods are expressed in terms of global analytic quantities, including \(B\), \(L\), and the Fourier norms of the smoothly truncated polynomial and its derivatives, together with the output norm. Their dependence on the polynomial degree is implicit through these parameters. These quantities may themselves grow with the degree; for example, \(f(z)=z^d\) satisfies \(B=R^d\) on the contour \(|z|=R>1\).
However, when these global variables are controlled, both of our methods have nothing to do with the degree, number of terms and coefficients of monomials of the polynomial, which is a great benefit for polynomials that perform well but with many terms with large coefficients that cancel each other out, such as Chebyshev polynomials. 

For general polynomials, even if their degree $d$ is not high, when $n$ is large, describing their behavior requires $\binom{n+d}{d}$ terms. If we simply use LCU to sum the block encodings of each term, the LCU coefficients $\|\boldsymbol c\|_1$ will be very large, even if this polynomial sometimes performs well on the joint spectrum of the matrices we are interested in. Similarly, our method can overcome this problem.

If matrices are restricted to commuting Hermitian, the Chebyshev/QSP construction~\cite{bornsweil2023mqet} provides a different polynomial-based
approach. Its complexity depends on the polynomial degree $d$ and dimension $n$ by $d^n$, and its comparison with our Fourier method depends on the norms of the Chebyshev and Fourier expansions for the specific problem.
\subsection{Contour integral}
Consider the problem of implementing a general multivariable matrix polynomial. Suppose that our aim is $f(\boldsymbol A)$, while $f$ is a multivariable complex polynomial and $\boldsymbol{A}=\left(A_{1}, \cdots, A_{\mathrm{n}}\right) \in \mathcal{A}^{\mathrm{n}}$ is an $n$-tuple of commuting matrices. Suppose the spectral radius of $A_j$ is $\rho(A_j)$, then all the eigenvalues of $A_j$ are covered by the disk $\{z:|z|\leq\rho(A_j) \}$, and we can choose $\Gamma_j$ as the boundary of $\{z:|z|\leq\rho(A_j)+a_j\}$, with $l_j=2\pi(\rho(A_j)+a_j)$.

Assume that for any $j$, $A_j$ can be diagonalized, then linear algebra concludes that they can be diagonalized by the same invertible matrix. Let $A_j = SD_jS^{-1}$, each element on the diagonal of $zI-D_j$ is not less than $a_j$, then
\begin{equation}
    \left\|(z I -A_j)^{-1}\right\| = \left\|S(z I -D_j)^{-1}S^{-1}\right\| \leq  \left\|S\right\|\left\|S^{-1}\right\|\left\|(z I -D_j)^{-1}\right\| \leq \frac{\kappa_{S}}{a_j}, 
\end{equation}
where $\kappa_S\geq\|S\|\|S^{-1}\|$ is the condition number of $S$. Therefore, we can let $\gamma_j= \frac{\kappa_{S}}{a_j}$ so that according to Theorem \ref{theorem:contour_integral_complexity}, we need $$\mathcal{O}\left(\frac{B\prod_{j=1}^n2\pi(\rho(A_j)+a_j)\frac{\kappa_{S}}{a_j}}{\left\|f(\boldsymbol A)\ket{\psi}\right\|}\cdot\frac{\kappa_{S}}{a_j}(\rho(A_j)+\alpha_j)\mathrm{log}\frac{B \prod_{j=1}^n2\pi(\rho(A_j)+a_j)\frac{\kappa_{S}}{a_j}}{\|f(\boldsymbol A)\ket{\psi}\|\epsilon}\right)$$ queries to $U_{A_j}$, $$\mathcal{O}\left(\frac{B\prod_{j=1}^n2\pi(\rho(A_j)+a_j)\frac{\kappa_{S}}{a_j}}{\left\|f(\boldsymbol A)\ket{\psi}\right\|}\right) $$ queries to $U_{\psi}$ and $$\mathcal{O}\left(\sum_{j=1}^n\log\frac{n(B\gamma_j+BK_{\Gamma_j}+L)\prod_{j=1}^n 2\pi(\rho(A_j)+a_j)\frac{\kappa_{S}}{a_j}}{\|f(\boldsymbol A)\ket{\psi}\|\epsilon}\right)+\log n+a+3$$ ancilla qubits to implement $f(A)$. Because $\rho (A_j)\leq\|A_j\|\leq\alpha_j$, we can choose $a_j=\mathcal{O}(\alpha_j)$ so that we need $$\mathcal{O}\left(K_n\frac{B\kappa_{S}^{n+1}}{\left\|f(\boldsymbol A)\ket{\psi}\right\|}\cdot \mathrm{log}\frac{B \prod_{j=1}^n\kappa_{S}}{\|f(\boldsymbol A)\ket{\psi}\|\epsilon}\right)$$ queries to $U_{A_j}$, $$\mathcal{O}\left(K_n\frac{B\prod_{j=1}^n\kappa_{S}}{\left\|f(\boldsymbol A)\ket{\psi}\right\|}\right) $$ queries to $U_{\psi}$ and $$\mathcal{O}\left(\sum_{j=1}^n\log\frac{n(B\gamma_j+BK_{\Gamma_j}+L)\prod_{j=1}^n \kappa_{S}}{\|f(\boldsymbol A)\ket{\psi}\|\epsilon}\right)+\log n+a+3$$ ancilla qubits, where $K_n$ is a constant raised to the power of n. It is only related to the maximum value of $f$ in the disk, and has nothing to do with the degree of $f$ or the coefficient of the monomial in $f$. 

\subsection{Fourier transform}

In this part we suppose that our aim is $f(\boldsymbol H)$, while $f$ is a multivariable real polynomial and $\boldsymbol{H}=\left(H_{1}, \cdots, H_{\mathrm{n}}\right)$ is an $n$-tuple of commuting Hermitian matrices. Since for $j=1,\cdots,n$, $H_j$ is Hermitian, all of its eigenvalues are real numbers. To apply the Fourier transform approach to a polynomial \(f\), we first construct a smooth, integrable function that agrees with \(f\) on the joint spectrum of \(\boldsymbol H\). This is necessary because a nonzero polynomial is not integrable over \(\mathbb R^n\). Since \(f(\boldsymbol H)\) depends only on the values of \(f\) on the joint spectrum, we can modify \(f\) outside a neighborhood of this spectrum without changing the target operator. Assume that the absolute values of its eigenvalues are all less than $K - \xi$, then let
\begin{equation}
\phi(x) = 
\begin{cases}
1 & \text{for } |x| \leq K-\xi \\
\left[
1+\exp\left(
\frac{\xi}{K-|x|}
-\frac{\xi}{|x|-K+\xi}
\right)
\right]^{-1} & \text{for } K-\xi < |x| < K \\
0 & \text{for } |x| \geq K
\end{cases},
\end{equation}
and define the kernel function
\begin{equation}
    k(\boldsymbol{x}) = \prod_{i=1}^n \phi(x_i).
\end{equation}
Then due to that $f\in C^\infty(\mathbf{R}^n), k\in C_c^\infty(\mathbf{R}^n)$, we know that
\begin{equation}
    g(\boldsymbol{x}):=f(\boldsymbol{x})\cdot k(\boldsymbol{x})\in C_c^\infty(\mathbf{R}^n),
\end{equation}
and $g(\boldsymbol{x})$ is supported in $[-K,K]^n$. For any $\boldsymbol{x}$ with $|\boldsymbol{x}|_\infty\leq K-\xi$, $g(\boldsymbol{x})=f(\boldsymbol{x})$. Due to that $[-K+\xi,K-\xi]^n$ covers the joint spectrum of $\boldsymbol{H}$, we have $g(\boldsymbol{H})=f(\boldsymbol{H})$. According to \cref{theorem:fourier_transform_complexity}, we need $$\mathcal{O}\left(\left(\alpha_jM+\frac{\log (1 / \epsilon^\prime)}{\log (e+\log (1 / \epsilon^\prime) /(\alpha_jM))}\right)\cdot\frac{\int_{\mathbb{R}^n} \left|\widehat{g}(\boldsymbol{\lambda})\right| \mathrm{~d} \boldsymbol{\lambda}}{(2 \pi)^{\frac{n}{2}}\|f(\boldsymbol H)\ket{\psi}\|}\right)$$ queries to $U_{H_j}$, 
$$\mathcal{O}\left(\frac{\int_{\mathbb{R}^n} |\widehat{g}(\boldsymbol{\lambda})| \mathrm{~d} \boldsymbol{\lambda}}{(2 \pi)^{\frac{n}{2}}\|f(\boldsymbol H)\ket{\psi}\|}\right)$$ queries to $U_{\psi}$ and $n\log N+\log n+a+4$ ancilla qubits, where 
\begin{equation}
    \epsilon^\prime\leq\frac{ (2 \pi)^{\frac{n}{2}}\|f(\boldsymbol H)\ket{\psi}\|\epsilon}{2n\int_{\mathbb{R}^n} \left|\widehat{g}(\boldsymbol{\lambda})\right| \mathrm{~d} \boldsymbol{\lambda}},
\end{equation}
\begin{equation}
    M\geq\left(\frac{1}{(2 \pi)^{\frac{n}{2}}}\cdot\frac{8}{\|f(\boldsymbol H)\ket{\psi}\|\epsilon}\cdot\sum_{j=1}^n\left\|\left(\frac{\partial^kg}{\partial x_j^k}\right)^{\wedge}\right\|_{L^1(\mathbb{R}^n)}\right)^\frac{1}{k},
\end{equation}
and
\begin{equation}
    N\geq\frac{(2M)^{n+1}}{(2 \pi)^{n}}\cdot\left(\left\|g\right\|_{L_1} \sum_{j=1}^n\|H_j\|+ \sum_{j=1}^n\left\|x_jg\right\|_{L_1}\right)\cdot\frac{8}{\|f(\boldsymbol H)\ket{\psi}\|\epsilon}.
\end{equation}

{\small
\bibliographystyle{alpha}
\bibliography{reference}
}
\appendix
\section{Non-exponential dependence of Fourier transform}
\label{Appendix:Non-exponential dependence of Fourier transform}
In this section, we show how the results given by \cref{theorem:fourier_transform_complexity} behave when the properties of $f$ are good. 
Let $f(\boldsymbol{x})$ be the $n$-dimensional standard normal density:
$$
f(\boldsymbol{x}) = \frac{1}{(2 \pi)^{\frac{n}{2}}} \exp\Big(-\frac{\|\boldsymbol{x}\|^2}{2}\Big), \quad \boldsymbol{x} \in \mathbb{R}^n.
$$

Using the Fourier transform convention

$$
\widehat f(\boldsymbol{\lambda})=\frac{1}{(2 \pi)^{\frac{n}{2}}}\int_{\mathbb{R}^n} f(\boldsymbol{x})\mathrm{e}^{-\mathrm{i}\boldsymbol{\lambda}\cdot\boldsymbol{x}}\mathrm d\boldsymbol{x},
$$

we obtain
$$
\widehat f(\boldsymbol{\lambda})=\frac{1}{(2 \pi)^{\frac{n}{2}}}\exp\left(-\frac{\|\boldsymbol{\lambda}\|^2}{2}\right).
$$

The \(L^1\)-norm of \(\widehat f\) is
$$
\int_{\mathbb R^n} \left|\widehat f(\boldsymbol{\lambda})\right|\,\mathrm d\boldsymbol{\lambda}
=\frac{1}{(2 \pi)^{\frac{n}{2}}}\int_{\mathbb R^n} \mathrm{e}^{-\|\boldsymbol{\lambda}\|^2/2}\,\mathrm d\boldsymbol{\lambda}
=\frac{1}{(2 \pi)^{\frac{n}{2}}}\cdot(2\pi)^{\frac{n}{2}}=1,
$$
Therefore, if $\boldsymbol{H}$ is a tuple of $0$ matrix, then $f(\boldsymbol{H})=(2\pi)^{-\frac{n}{2}}\cdot I$, and $\|f(\boldsymbol H)\ket{\psi}\|=(2\pi)^{-\frac{n}{2}}$. Substituting the result, we have
$$\mathcal{O}\left(\frac{\int_{\mathbb{R}^n} \left|\widehat{f}(\boldsymbol{\lambda})\right| \mathrm{~d} \boldsymbol{\lambda}}{(2 \pi)^{\frac{n}{2}}\|f(\boldsymbol H)\ket{\psi}\|}\right)=\mathcal{O}(1),
$$
this means the query complexity of $U_\psi$ is a constant, and
$$
\frac{1}{\epsilon^\prime}=\mathcal{O}\left(\frac{2n\int_{\mathbb{R}^n} \left|\widehat{f}(\boldsymbol{\lambda})\right| \mathrm{~d} \boldsymbol{\lambda}}{ (2 \pi)^{\frac{n}{2}}\|f(\boldsymbol H)\ket{\psi}\|\epsilon}\right)=\mathcal{O}\left(\frac{n}{\epsilon}\right).
$$
The differentiation property under this convention is
$$
\left(\frac{\partial^k f}{\partial x_j^k}\right)^\wedge(\boldsymbol{\lambda})
=(i\lambda_j)^k \widehat f(\boldsymbol{\lambda})
=\frac{(i\lambda_j)^k}{(2 \pi)^{\frac{n}{2}}} \mathrm{e}^{-\|\boldsymbol{\lambda}\|^2/2}.
$$
Hence
$$
\left\|\left(\frac{\partial^k f}{\partial x_j^k}\right)^\wedge\right\|_{L^1}
=\frac{1}{(2 \pi)^{\frac{n}{2}}}\int_{\mathbb R^n} |\lambda_j|^k \mathrm{e}^{-\|\boldsymbol{\lambda}\|^2/2}\mathrm d\boldsymbol{\lambda}.
$$

By symmetry,
$$
\sum_{j=1}^n \left\|\left(\frac{\partial^k f}{\partial x_j^k}\right)^\wedge\right\|_{L^1}
= \frac{n}{(2 \pi)^{\frac{n}{2}}}\int_{\mathbb R^n} |\lambda_1|^k \mathrm{e}^{-\|\boldsymbol{\lambda}\|^2/2}\,\mathrm d\boldsymbol{\lambda}.
$$
Using separability,
$$
\int_{\mathbb R^n} |\lambda_1|^k \mathrm{e}^{-\|\boldsymbol{\lambda}\|^2/2}\,\mathrm d\boldsymbol{\lambda}
=\left(\int_{\mathbb R} |\lambda|^k \mathrm{e}^{-\lambda^2/2}\,\mathrm d\lambda\right)
\left(\int_{\mathbb R} \mathrm{e}^{-\lambda^2/2}\,\mathrm d\lambda\right)^{\,n-1}.
$$
Now
$$
\int_{\mathbb R} \mathrm{e}^{-\lambda^2/2}\,\mathrm d\lambda=\sqrt{2\pi},
\qquad
\int_{\mathbb R} |\lambda|^k \mathrm{e}^{-\lambda^2/2}\,\mathrm d\lambda
=2\int_0^\infty \lambda^k \mathrm{e}^{-\lambda^2/2}\,\mathrm d\lambda
=2^{(k+1)/2}\,\Gamma\left(\frac{k+1}{2}\right).
$$
Combining these yields
$$
\sum_{j=1}^n \left\|\left(\frac{\partial^k f}{\partial x_j^k}\right)^\wedge\right\|_{L^1}
= \frac{n}{(2 \pi)^{\frac{n}{2}}}\cdot 2^{(k+1)/2}\Gamma\left(\frac{k+1}{2}\right)\cdot(2\pi)^{\frac{n-1}{2}}
$$
For $k=1$, we have

$$
\Gamma(1)=1, 
\sum_{j=1}^n \left\|\left(\frac{\partial^k f}{\partial x_j^k}\right)^\wedge\right\|_{L^1}
= n\sqrt{\frac{2}{\pi}}.
$$

Still assume that $\boldsymbol{H}$ is a tuple of $0$ matrix, then $\|f(\boldsymbol H)\ket{\psi}\|=(2\pi)^{-n/2}$ and the condition that $M$ needs to satisfy is:
$$
    M\geq\frac{1}{(2 \pi)^{\frac{n}{2}}}\cdot\frac{8}{\|f(\boldsymbol H)\ket{\psi}\|\epsilon}\cdot\sum_{j=1}^n\left\|\left(\frac{\partial f}{\partial x_j}\right)^{\wedge}\right\|_{L^1(\mathbb{R}^n)}=\mathcal{O}\left(\frac{n}{\epsilon}\right).
$$

Therefore, when we estimate c using the intermediate inequality $\|\boldsymbol c\|_1\lesssim \frac{1}{(2 \pi)^{\frac{n}{2}}} \int_{\mathbb{R}^n} \left|\widehat{f}(\boldsymbol{\lambda})\right| \mathrm{~d} \boldsymbol{\lambda}$, the query complexity of $U_{H_j}$ is
$$\mathcal{O}\left(\left(\alpha_jM+\frac{\log (1 / \epsilon^\prime)}{\log (e+\log (1 / \epsilon^\prime) /(\alpha_jM))}\right)\frac{\int_{\mathbb{R}^n} \left|\widehat{f}(\boldsymbol{\lambda})\right| \mathrm{~d} \boldsymbol{\lambda}}{(2 \pi)^{\frac{n}{2}}\|f(\boldsymbol H)\ket{\psi}\|}\right)=\widetilde{\mathcal{O}}\left(\frac{\alpha_jn}{\epsilon}\right)$$ 
which means the query complexity of $U_{H_j}$ depends approximately linearly on $n$. Since
$$
\sum_{j=1}^n \sup_{\boldsymbol{\lambda}}\left|\frac{\partial}{\partial \lambda_j}\widehat f(\boldsymbol{\lambda})\right|
=\sum_{j=1}^n\frac{1}{(2\pi)^{\frac{n}{2}}}\sup_{\boldsymbol{\lambda}}\left(|\lambda_j|\mathrm{e}^{-\|\boldsymbol{\lambda}\|^2/2}\right)=\frac{n}{(2\pi)^{\frac{n}{2}}}\mathrm{e}^{-\frac{1}{2}},
$$
we can use the conclusion given in \cref{equation:fourier_discretization_error} to get the bound of $N$ and further bound the number of ancilla qubits, the condition of $N$ is

$$
\begin{aligned}
N\geq\frac{(2M)^{n+1}}{(2 \pi)^{\frac{n}{2}}}\left(\sup_{\boldsymbol{\lambda}}\left|\widehat{f}(\boldsymbol{\lambda})\right| \sum_{j=1}^n\|H_j\|+ \sum_{j=1}^n\sup_{\boldsymbol{\lambda}}\left|\frac{\partial}{\partial \lambda_j}\widehat{f}(\boldsymbol{\lambda})\right|\right)\cdot\frac{8}{\|f(\boldsymbol H)\ket{\psi}\|\epsilon}=\frac{8n(2M)^{n+1}}{(2\pi)^{\frac{n}{2}}\mathrm{e}^{\frac{1}{2}}\epsilon},
\end{aligned}
$$
and the amount of ancilla qubits $n\log N+\log n+a+4=\mathcal{O}(n^2\log\frac{n}{\epsilon}+a)$
which is polynomial, rather than exponential, in the dimension \(n\).
\end{document}